\documentclass[11pt]{article}

\usepackage{xspace}
\usepackage{url}
\usepackage{enumitem}
\usepackage{color}
\usepackage{algorithm,algorithmicx}
\usepackage[noend]{algpseudocode}
\usepackage{multirow}
\usepackage{hyperref}
\usepackage{caption}
\usepackage{listings}
\usepackage{courier}
\usepackage{amsthm}
\usepackage{ulem}

\usepackage[english]{babel}
\usepackage{url}
\usepackage[utf8x]{inputenc}
\usepackage{amsmath}
\usepackage{graphicx}
\graphicspath{{images/}}
\usepackage{parskip}
\usepackage{fancyhdr}
\usepackage{vmargin}
\usepackage{multicol}

\setmarginsrb{3 cm}{2.5 cm}{3 cm}{1.5 cm}{0 cm}{1 cm}{0 cm}{1 cm}

\title{Accelerating Dynamic Graph Analytics on GPUs}	
\author{M. Sha et al.}
\date{July 15, 2017}

\makeatletter
\let\thetitle\@title
\let\theauthor\@author
\let\thedate\@date
\makeatother

\begin{document}
\pagestyle{empty}

\begin{titlepage}
	\centering
    \vspace*{0.5 cm}
    \includegraphics[scale = 0.75]{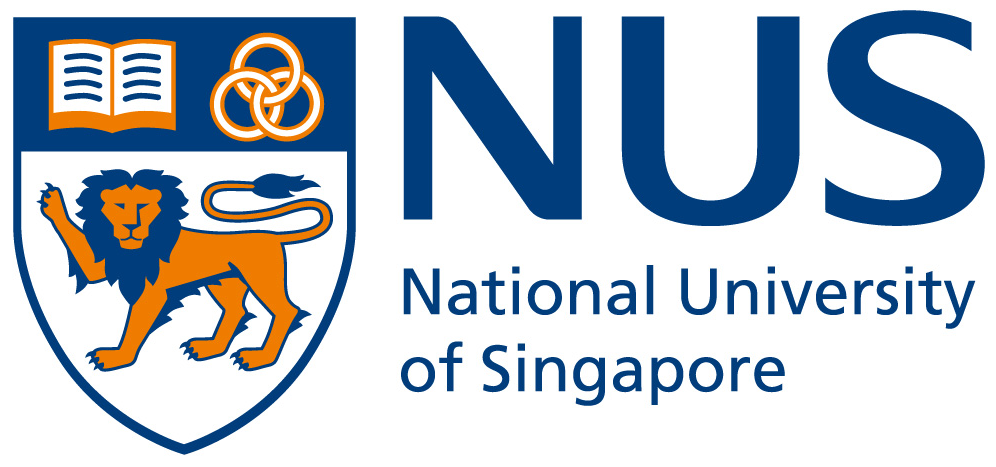}\\[1.0 cm]	

	\rule{\linewidth}{0.2 mm} \\[0.4 cm]
	{ \Huge \bfseries \thetitle}\\
	\rule{\linewidth}{0.2 mm} \\[1.5 cm]

	\textbf{\LARGE Technical Report} \\
	\vspace{2mm}
	\textbf{\large Version 2.1} \\
	\vspace{1cm}
	{\Large Mo Sha, Yuchen Li, Bingsheng He and Kian-Lee Tan} \\
	
	\vspace{9cm}
	{\large \thedate}\\[2 cm]
 
	\vfill
	
\end{titlepage}


\tableofcontents
\pagebreak

\pagestyle{fancy}

\newtheorem{definition}{Definition}
\newtheorem{lemma}{Lemma}
\newtheorem{example}{Example}
\newtheorem{theorem}{Theorem}

\newcommand{\system}{\texttt{SAG\textsuperscript{2}E}\xspace}
\newcommand{\pma}{\texttt{PMA}\xspace}
\newcommand{\gpma}{\texttt{GPMA}\xspace}
\newcommand{\gpmaplus}{\texttt{GPMA+}\xspace}
\newcommand{\stinger}{\texttt{Stinger}\xspace}
\newcommand{\adjlists}{\texttt{AdjLists}\xspace}
\newcommand{\gpucsr}{\texttt{cuSparseCSR}\xspace}

\newcommand{\todo}[1]{\textcolor{blue}{\textbf{?-}{#1}\textbf{-?}}}

\newcommand{\dsreddit}{{\tt Reddit}\xspace}
\newcommand{\dspokec}{{\tt Pokec}\xspace}
\newcommand{\dsrandom}{{\tt Random}\xspace}
\newcommand{\dsgraphlib}{{\tt Graph500}\xspace}
\newcommand{\appagerank}{{\tt PageRank}\xspace}
\newcommand{\apcc}{{\tt ConnectedComponent}\xspace}
\newcommand{\apbfs}{{\tt BFS}\xspace}
\newcommand{\apspanner}{{\tt Spanner}\xspace}
\newcommand{\aptcm}{{\tt TCM}\xspace}
\newcommand{\apagg}{{\tt SubGAgg}\xspace}
\newcommand{\aphybrid}{{\tt TCM+SubGAgg}\xspace}

\newcommand{\spmv}{{\tt SpMV}\xspace}

\newcommand{\vset}{V}
\newcommand{\eset}{E}

\algdef{SE}[SUBALG]{Indent}{EndIndent}{}{\algorithmicend\ }
\algtext*{Indent}
\algtext*{EndIndent}

\algdef{SE}[DOWHILE]{Do}{doWhile}{\algorithmicdo}[1]{\algorithmicwhile\ #1}

\newcommand{\highlight}[1]{\emph{\textbf{\uuline{#1}}}}

\newcommand{\marked}[1]{{\color{black} #1}}
\setcounter{page}{1}
\setcounter{figure}{0}
\begin{abstract}
As graph analytics often involves compute-intensive operations,
GPUs have been extensively used to accelerate the processing.
However, in many applications
such as social networks, cyber security, and fraud detection, their representative graphs evolve frequently and
one has to perform a rebuild of the graph structure on GPUs to incorporate the updates.
Hence, rebuilding the graphs becomes the bottleneck of processing high-speed graph streams.
In this paper, 
we propose a GPU-based dynamic graph storage scheme to support existing graph algorithms easily.
Furthermore,
we propose parallel update algorithms to support efficient stream updates
so that the maintained graph is immediately available for high-speed analytic processing on GPUs.
Our extensive experiments with three streaming applications
on large-scale real and synthetic datasets demonstrate the superior performance of our proposed approach.
\end{abstract} 
\section{Introduction}\label{sec:intro}
Due to the rising complexity of data generated in the big data era,
graph representations are used ubiquitously. 
Massive graph processing has emerged as the de facto standard 
of analytics on web graphs, social networks (e.g., Facebook and Twitter), 
sensor networks (e.g., Internet of Things)
and many other application domains which involve high-dimen-sional data (e.g., recommendation systems). 
These graphs are often highly dynamic: network traffic data averages $10^9$ packets/hour/router
for large ISPs \cite{Guha:2012:GSS}; Twitter has $500$ million tweets per day \cite{TweetStats}.
Since real-time analytics is fast becoming the norm \cite{Iyer:2016:TGP,Braun:2015:AMH,McGregor:2014:GSA,stonebraker20058},
it is thus critical for operations on dynamic massive graphs to be processed efficiently.  

Dynamic graph analytics has a wide range of applications. 
Twitter can recommend information based on the
up-to-date TunkRank (similar to PageRank) computed
based on a dynamic attention graph \cite{Cheng:2012:KTP} and cellular
network operators can fix traffic hotspots in their
networks as they are detected~\cite{Iyer:2015:RTC}.
To achieve real-time performance, there is a growing
interest to offload the graph analytics to GPUs due to
its much stronger arithmetical power and higher memory bandwidth compared with CPUs~\cite{stratton2012optimization}.
Although existing solutions, e.g. Medusa \cite{Zhong:2014:SGP} and Gunrock \cite{Wang:2015:GHG}, have
explored GPU graph processing,
we are aware that the only one work~\cite{King:2016:DSM}
has considered a dynamic graph scenario which
is a major gap for running analytics on GPUs. 
In fact, a delay in updating a dynamic graph may lead to undesirable consequences.
For instance, consider an online travel insurance system that detects potential frauds by running ring analysis on profile graphs
built from active insurance contracts \cite{Akoglu:2015:GBA}.
Analytics on an outdated profile graph may fail to detect frauds which can cost millions of dollars.
However, updating the graph will be too slow for issuing contracts and processing claims in real time,
which will severely influence legitimate customers' user experience.
This motivates us to develop an update-efficient graph structure on GPUs to support dynamic graph analytics.
%

There are two major concerns
when designing a GPU-based dynamic graph storage scheme.
First, the proposed storage scheme should handle both insertion and deletion operations efficiently.
Though processing updates against insertion-only graph stream could be handled by reserving extra spaces to accommodate the updates,
this na\"{\i}ve approach fails to preserve the locality of the graph entries and cannot support deletions efficiently. 
Considering a common sliding window model on a graph edge stream,
each element in the stream is an edge in a graph and analytic tasks are performed on the graph induced by all edges in the up-to-date window \cite{ywang, crouch2013dynamic, datar2002maintaining}.
A na\"{\i}ve approach needs to access the entire graph in the sliding window to process deletions.
This is obviously undesirable against high-speed streams. 
Second, the proposed storage scheme should be general enough for supporting existing graph formats on GPUs
so that we can easily reuse existing static GPU graph processing solutions for graph analytics.
Most large graphs are inherently \emph{sparse}. 
To maximize the efficiency, 
existing works \cite{Ashari:2014:FSM,Liu:2016:ICB,Lin:2014:ESM,King:2016:DSM,Yang:2011:FSM} on GPU sparse graph processing
rely on optimized data formats and arrange
the graph entries in certain sorted order, e.g. CSR \cite{Liu:2016:ICB,Ashari:2014:FSM} sorts the entries by their row-column ids.
However, to the best of our knowledge, no schemes on GPUs can support efficient updates and maintain a sorted graph format
at the same time, other than a rebuild. 
This motivates us to design an update-efficient sparse graph storage scheme on GPUs
while keeping the locality of the graph entries for processing massive analytics instantly. 


In this paper, we introduce a GPU-based dynamic graph analytic framework followed by proposing the dynamic graph storage scheme on GPUs.
Our preliminary study shows that a cache-oblivious data structure, i.e., Packed Memory Array (\pma \cite{Bender:2005:COB,Bender:2007:PMA}), 
can potentially be employed for maintaining dynamic graphs on GPUs.
\pma, originally designed for CPUs~\cite{Bender:2005:COB,Bender:2007:PMA}, maintains sorted elements in a partially contiguous fashion by leaving gaps to accommodate fast updates
with a constant bounded gap ratio.
The simultaneously sorted and contiguous characteristic of \pma nicely fits the scenario of GPU streaming graph maintenance.
However, the performance of \pma degrades when updates occur in locations which are close to each other, due to the unbalanced utilization of reserved spaces. 
Furthermore, as streaming updates often come in batches rather than one single update at a time, \pma does not support parallel insertions and 
it is non-trivial to apply \pma to GPUs due to its intricate update patterns
which may cause serious thread divergence and uncoalesced memory access issues on GPUs. 

We thus propose two GPU-oriented algorithms, i.e. \gpma and \gpmaplus, to support efficient parallel batch updates.
\gpma explores a lock-based approach which becomes increasingly
popular due to the recent GPU architectural evolution for supporting atomic operations \cite{davidson2014work, kaleem2016synchronization}.
While \gpma works efficiently for the case where few concurrent updates conflict, e.g., small-size update batches with random updating edges in each batch,  
there are scenarios where massive conflicts occur and hence, we propose a lock-free approach, i.e. \gpmaplus. 
Intuitively, \gpmaplus is a bottom-up approach by prioritizing updates that occur in similar positions.
The update optimizations of our proposed \gpmaplus are able to maximize coalesced memory access
and achieve linear performance scaling w.r.t the number of computation units on GPUs,
regardless of the update patterns.

%
%

In summary, the key contributions of this paper are the following:
\begin{itemize}[leftmargin=*, noitemsep]
\item We introduce a framework for GPU dynamic graph analytics and propose, the first of its kind, a GPU dynamic graph storage scheme to
pave the way for real-time dynamic graph analytics on GPUs. 
\item We devise two GPU-oriented parallel algorithms: \gpma and \gpmaplus,
to support efficient updates against high-speed graph streams.  
%
\item We conduct extensive experiments to show the performance superiority of \gpma and \gpmaplus. 
In particular, we design different update patterns on real and synthetic graph streams
to validate the update efficiency of our proposed algorithms against their CPU counterparts as well as the GPU rebuild baseline.
In addition, 
we implement three real world graph analytic applications on the graph streams to demonstrate 
the efficiency and broad applicability of our proposed solutions.
In order to support larger graphs, we extend our proposed formats to multiple GPUs and demonstrate the scalability of our approach with multi-GPU systems. 
\end{itemize}

The remainder of this paper is organized as follows. 
The related work is discussed in Section~\ref{sec:literature}. 
Section~\ref{sec:overview} presents a general workflow of dynamic graph processing on GPUs.
Subsequently, we describe \gpma and \gpmaplus in Sections~\ref{sec:gpma}-\ref{sec:gpmaplus} respectively. 
Section~\ref{sec:experiment} reports results of a comprehensive experimental evaluation. 
We conclude the paper and discuss some future works in Section~\ref{sec:conclusion}.

\section{Related Work}\label{sec:literature}
In this section, we review related works in three different categories as follows.
\subsection{Graph Stream Processing}
Over the last decade, there has been an immense interest in designing efficient algorithms
for processing massive graphs in the data stream model (see \cite{McGregor:2014:GSA} for a detailed survey). 
This includes the problems of PageRank-styled scores \cite{ohsaka2015efficient}, connectivity \cite{Feigenbaum:2005:OGP}, spanners \cite{lkin:2011:SFD},
counting subgraphs e.g. triangles \cite{Tsourakakis:2009:CTM} and summarization \cite{Tang:2016:GSS}. 
However, these works mainly focus on the theoretical study to achieve the best approximation solution
with linear bounded space. 
Our proposed methods can incorporate existing graph stream algorithms with ease 
as our storage scheme can support most graph representations used in existing algorithms.

Many systems have been proposed for streaming data processing, 
e.g. Storm \cite{Toshniwal:2014:STO}, Spark Streaming \cite{Zaharia:2013:DSF}, Flink \cite{Flink}.
Attracted by its massively parallel performance, 
several attempts have successfully demonstrated the advantages 
of using GPUs to accelerate data stream processing \cite{Verner:2012:SPR,Zhang:2011:GStream}.
However, the aforementioned systems focus on general stream processing and lack support for graph stream processing.
Stinger \cite{Ediger:2012ve} is a parallel solution to support dynamic graph analytics on a single machine.   
More recently, Kineograph \cite{Cheng:2012:KTP}, CellIQ \cite{Iyer:2015:RTC} and GraphTau \cite{Iyer:2016:TGP} 
are proposed to address the need for general time-evolving graph processing under the distributed settings. 
\marked{However, to our best knowledge, existing works focusing on CPU-based time-evolving graph processing will be inefficient on GPUs, because CPU and GPU are two architectures with different design principles and performance concerns in the parallel execution.} We are aware that only one work~\cite{King:2016:DSM} explores the direction of using GPUs to process real-time analytics on dynamic graphs. However, it only supports insertions and lacks an efficient indexing mechanism.

\subsection{Graph Analytics on GPUs}\label{sec:literature:analytics}
Graph analytic processing is inherently data- and compute-intensive. 
Massively parallel GPU accelerators are powerful to achieve supreme performance of many applications. 
Compared with CPU, which is a general-purpose processor featuring large cache size and high single core processing capability, 
GPU devotes most of its die area to a large number of simple Arithmetic Logic Units (ALUs),
and executes code in a SIMT (Single Instruction Multiple Threads) fashion.
With the massive amount of ALUs, GPU offers orders of magnitude higher computational throughput than CPU in applications with 
ample parallelism. 
This leads to a spectrum of works which explore the usage of GPUs to accelerate graph analytics and demonstrate
immense potentials. Examples include breath-first search (BFS) \cite{Liu:2016:ICB}, 
subgraph query \cite{Lin:2014:ESM}, PageRank \cite{Ashari:2014:FSM} and many others. 
The success of deploying specific graph algorithms on GPUs motivates 
the design of general GPU graph processing systems like Medusa \cite{Zhong:2014:SGP} and Gunrock \cite{Wang:2015:GHG}. 
However, the aforementioned GPU-oriented graph algorithms and systems assume 
static graphs. 
To handle dynamic graph scenario, existing works have to perform a rebuild on GPUs against each single update.
DCSR \cite{King:2016:DSM} is the only solution, to the best of our knowledge, which is designed for insertion-only scenarios
as it is based on linked edge block and rear appending technique.
However, it does not support deletions or efficient searches.
We propose \gpma to enable efficient dynamic graph updates (i.e. insertions and deletions) on GPUs in a fine-grained manner. 
In addition,  existing GPU-optimized graph analytics and systems can replace their storage layers directly with ease since the fundamental graph storage schemes used in
existing works can be directly implemented on top of our proposed storage scheme. 

\subsection{Storage Formats on GPUs}
Sparse matrix representation is a popular choice for storing large graphs on GPUs~\cite{cuSparse, cusp, Zhong:2014:SGP, Wang:2015:GHG}
The Coordinate Format \cite{dang2012sliced} (COO) is the simplest format which only stores non-zero matrix entries by their coordinates with values. 
COO sorts all the non-zero entries by the entries' row-column key for fast entry accesses. 
CSR~\cite{Liu:2016:ICB,Ashari:2014:FSM} compresses COO's row indices into an offset array to reduce the memory bandwidth when accessing the sparse matrix. 
To optimize matrices with different non-zero distribution patterns, 
many customized storage formats were proposed, e.g., Block COO~\cite{yan2014yaspmv} (BCCOO), Blocked Row-Column~\cite{ashari2014efficient} (BRC) and Tiled COO~\cite{yang2011fast} (TCOO). 
Existing formats require to maintain a certain sorted order of their storage base units according to the unit's position
in the matrix, e.g. entries for COO and blocks for BCCOO, and still ensure the locality of the units. 
As mentioned previously, few prior schemes can handle
efficient sparse matrix updates on GPUs. 
To the best of our knowledge, \pma \cite{Bender:2005:COB,Bender:2007:PMA} is a common structure which maintains a sorted array in a contiguous manner and supports
efficient insertions/deletions. However, \pma is designed for CPU and no concurrent updating algorithm is ever proposed. 
Thus, we are motivated to propose \gpma and \gpmaplus for supporting efficient concurrent updates
on all existing storage formats.
\section{A dynamic framework on GPUs} \label{sec:overview}

\begin{figure}[t!]\centering
	\includegraphics[width=0.6\linewidth]{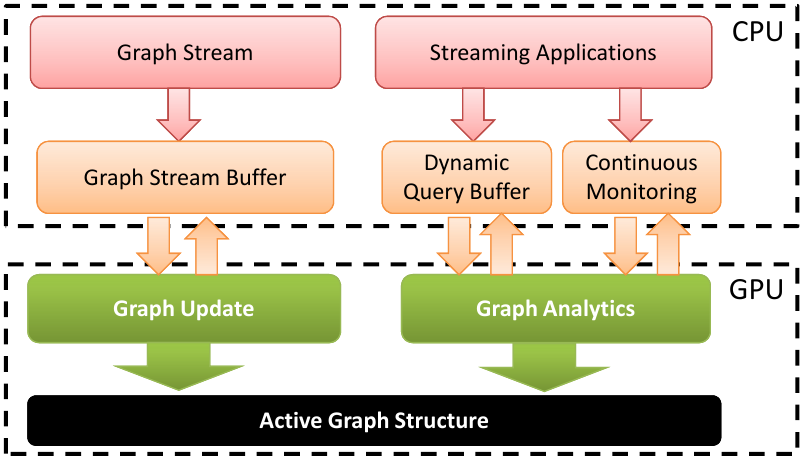}
	\caption{The dynamic graph analytic framework}
	\label{fig:framework}
\end{figure}

To address the need for real-time dynamic graph analytics, 
we offload the tasks of concurrent dynamic graph maintenance
and its corresponding analytic processing to GPUs. 
In this section, we introduce a general GPU dynamic graph analytic framework.
The design of the framework takes into account two major concerns: the framework should not only handle graph updates efficiently but also support existing GPU-oriented graph analytic algorithms 
without forfeiting their performance.

\noindent\textbf{Model:}
We adopt a common sliding window graph stream model \cite{McGregor:2014:GSA,Iyer:2015:RTC,Tang:2016:GSS}. 
The sliding window model consists of an unbounded sequence of elements $(u,v)_t$ \footnote{Our framework handles both directed and undirected edges.}
which indicates the edge $(u,v)$ arrives at time $t$,
and a sliding window which keeps track of the most recent edges. 
As the sliding window moves with time, new edges in the stream keep being inserted into the window and expiring edges are deleted. 
In real world applications, the sliding window of a graph stream can be used to
monitor and analyze fresh social actions that appearing on Twitter \cite{ywang} or
the call graph formed by the most recent CDR data \cite{Iyer:2015:RTC}. 
In this paper, we focus on presenting how to handle edge streams but
our proposed scheme can also handle the dynamic \emph{hyper graph} scenario with hyper edge streams. 

Apart from the sliding window model, the graph stream model which involves explicit insertions and deletions (e.g., a user requests to add or delete a friend in the social network)
is also supported by our scheme as the proposed dynamic graph storage structure is designed to handle random update operations.
That is, our system supports two kinds of updates, \emph{implicit} ones generated from the sliding window mechanism and \emph{explicit} ones generated from upper level applications or users.

%

The overview of the dynamic graph analytic framework is presented in Figure~\ref{fig:framework}.
Given a graph stream, there are two types of streaming tasks supported by our framework. The first type is the ad-hoc queries
such as neighborhood and reachability queries on the graph which is constantly changing. 
The second type is the monitoring tasks like tracking PageRank scores.
We present the framework by illustrating how to handle the graph streams and the corresponding queries 
while hiding data transfer between CPU and GPU, as follows:

\begin{figure}[t!]\centering
	\includegraphics[width=0.6\linewidth]{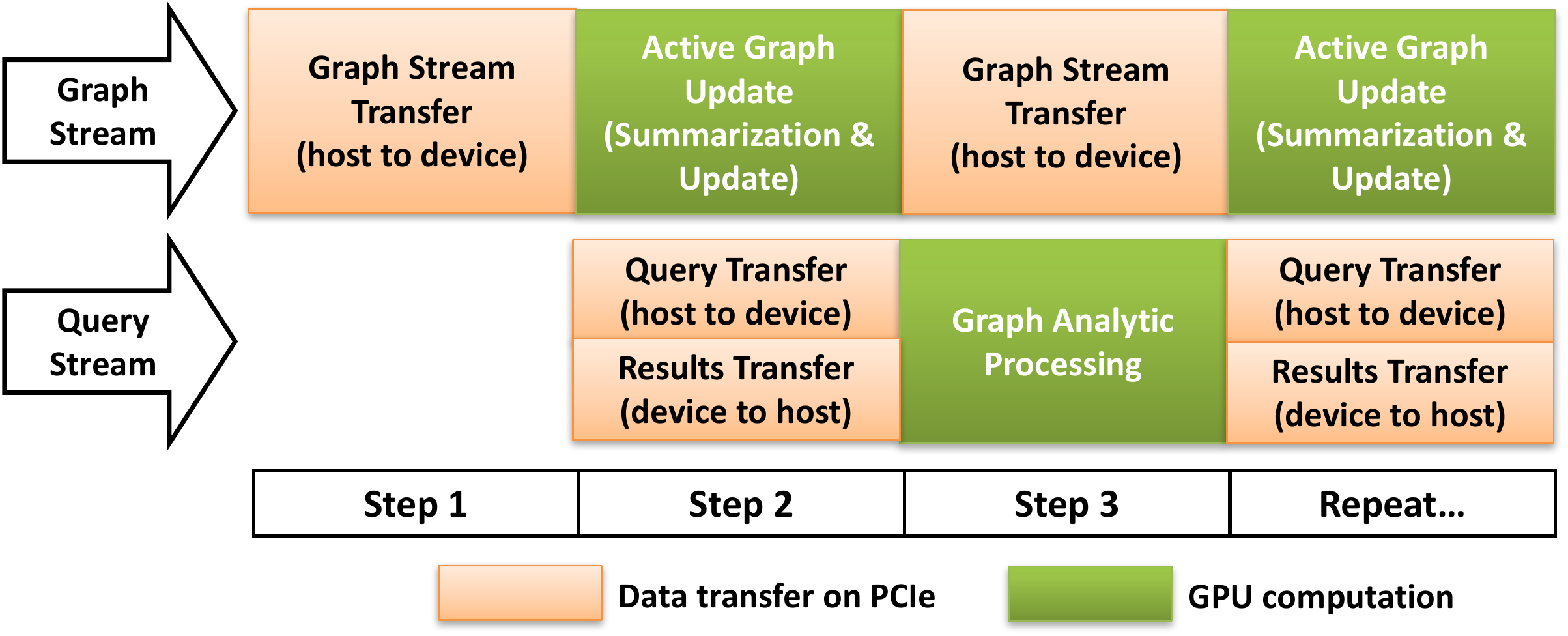}
	\caption{Asynchronous streams}
	\label{fig:pipeline}
\end{figure}

\begin{figure*}[t]\centering
	\hspace{-6em}
	\begin{minipage}[c]{0.5\linewidth}
		\includegraphics[width=\linewidth]{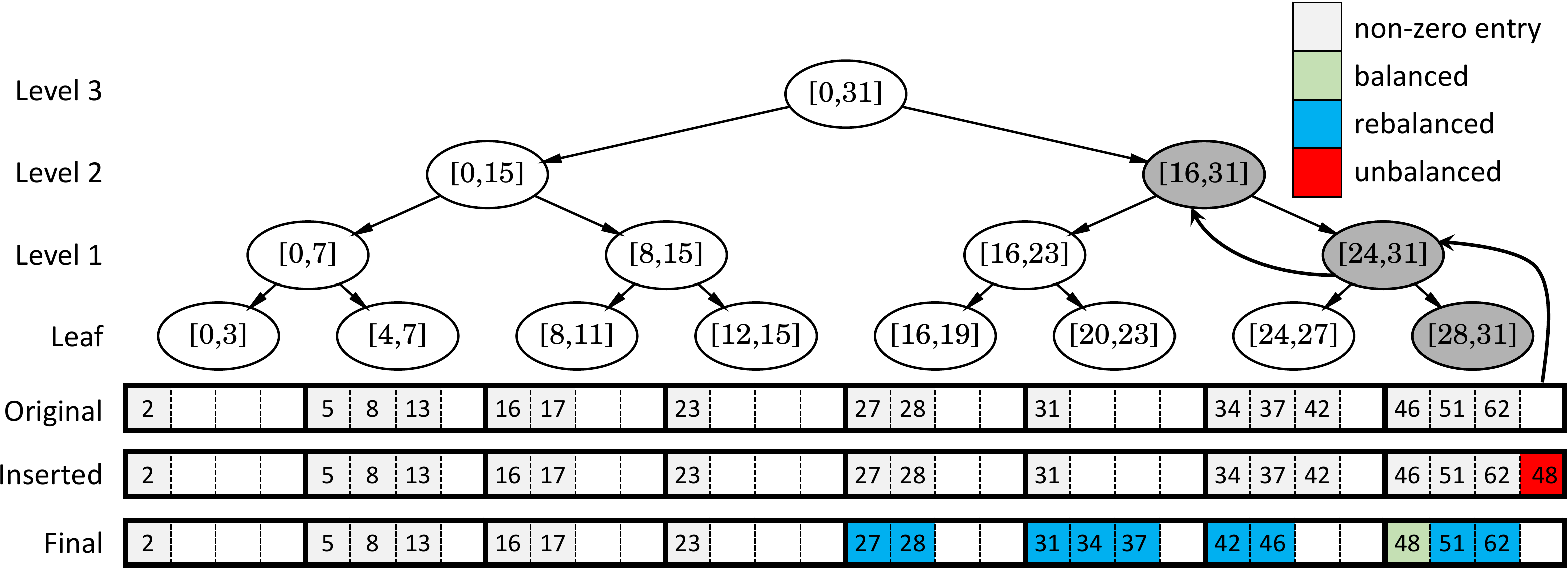}
	\end{minipage}
	\hspace{2mm}
	\begin{minipage}{0.40\linewidth}
		\centering
		\scriptsize
		\begin{tabular}[c]{|c|c|c|c|c|c|c|}
			\hline
			& Leaf & Level 1 & Level 2 & Level 3 \\ \hline
			segment size & 4 & 8 & 16 & 32 \\ \hline
			density lower bound $\rho$ & 0.08 & 0.19 & 0.29 & 0.40 \\ \hline
			density upper bound $\tau$ & 0.92 & 0.88 & 0.84 & 0.80 \\ \hline
			min \# of entries & 1 & 2 & 4 & 8 \\ \hline
			max \# of entries & 3 & 6 & 12 & 24 \\ \hline
		\end{tabular}\\
	\end{minipage}
	\caption{\pma insertion example (Left: PMA for insertion; Right: predefined thresholds)}
	\label{fig:pma_insert_example}
\end{figure*}

\vspace{1mm}
\noindent\textbf{Graph Streams:}
The graph stream buffer module batches the incoming graph streams on the CPU side (host) and 
periodically sends the updating batches to the graph update module located on GPU (device). 
The graph update module updates the ``active'' graph stored on the device by using the batch received. 
The ``active'' graph is stored in the format of our proposed GPU dynamic graph storage structure.
The details of the graph storage structure and how to update the graph efficiently on GPUs will be 
discussed extensively in later sections.

\vspace{1mm}
\noindent\textbf{Queries:}
Like the graph stream buffer, the dynamic query buffer module batches ad-hoc queries submitted against 
the stored active graph, e.g., queries to check the dynamic reachability between pairs of vertices.   
The tracking tasks will also be registered in the continuous 
monitoring module, e.g., tracking up-to-date PageRank.
All ad-hoc queries and monitoring tasks will be transferred to the graph analytic module for GPU accelerated processing.
The analytic module interacts with the active graph to process the queries and the tracking tasks. 
Subsequently, the query results will be transferred back to the host.
As most existing GPU graph algorithms use optimized array formats like CSR to accelerate the performance~\cite{davidson2014work, kaleem2016synchronization, martone2010use, yang2011fast}, 
our proposed storage scheme provides an interface for storing the array formats. 
In this way, existing algorithms can be integrated into the analytic module with ease. 
We describe the details of the integration in Section~\ref{sec:generality}.

\vspace{1mm}
\noindent\textbf{Hiding Costly PCIe Transfer:}
Another critical issue on designing GPU-oriented systems is to minimize the data transfer between 
the host and the device through PCIe. 
Our proposed batching approach allows overlapping data
transfer by concurrently running analytic tasks on the device. 
Figure~\ref{fig:pipeline} shows a simplified schedule with two 
asynchronous streams: graph streams and query streams respectively. 
The system is initialized at Step $1$ where the batch containing incoming graph stream elements
is sent to the device. 
At Step $2$, while PCIe handles bidirectional data transfer for previous query results (device to host) and
freshly submitted query batch (host to device), the graph update module updates the active graph stored on the device. 
At Step $3$, the analytic module processes the received query batch on the device and 
a new graph stream batch is concurrently transferred from the host to the device. 
It is clear to see that, by repeating the aforementioned process,
all data transfers are overlapped with concurrent device computations.

\marked{
\section{{\Large \gpma} Dynamic Graph Processing}\label{sec:gpma}

To support dynamic graph analytics on GPUs, there are two major challenges discussed in the introduction.
The first challenge is to maintain the dynamic graph storage in the device memory of GPUs for efficient update as well as compute.
The second challenge is that the storage strategy should show its good compatibility with existing graph analytic algorithms on GPUs.

In this section, we discuss how to address the challenges with our proposed scheme.
First, we introduce \gpma for GPU resident graph storage to simultaneously achieve update and compute efficiency (Section~\ref{subsec:gpma}). 
Subsequently, we illustrate \gpma's generality in terms of deploying existing GPU based graph analytic algorithms (Section~\ref{sec:generality}).


\subsection{{\Large \gpma} Graph Storage on GPUs}\label{subsec:gpma}
In this subsection, we first discuss the design principles our proposed dynamic graph storage should follow. Then we introduce how to implement our proposal.

\noindent\textbf{Design Principles.} The proposed graph storage on GPUs should take into account the following principles:
\begin{itemize}[noitemsep,leftmargin=*]
\item The proposed dynamic graph storage should efficiently support a broad range of updating operations, including \emph{insertions}, \emph{deletions} and \emph{modifications}. Furthermore, it should have a good locality to accommodate the highly parallel memory access characteristic of GPUs, in order to achieve high memory efficiency.
\item The physical storage strategy should support common logical storage formats and the existing graph analytic solutions on GPUs based on such formats can be adapted easily. 
\end{itemize}
}
\noindent\textbf{Background of PMA.}
\gpma is primarily motivated by a novel structure, Packed Memory Array (\pma \cite{Bender:2005:COB,Bender:2007:PMA}),
which is proposed to maintain sorted elements in a partially continuous fashion
by leaving gaps to accommodate fast updates with a bounded gap ratio. 
\pma is a self-balancing binary tree structure.
Given an array of $N$ entries,
\pma separates the whole memory space into \emph{leaf segments} with $O(\log N)$ length and defines \emph{non-leaf segments} as the space occupied by their descendant segments. 
For any segment located at height $i$ (leaf height is $0$), 
\pma designs a way to assign the lower and upper bound density thresholds for the segment as $\rho_i$ and $\tau_i$ respectively to achieve $O(\log^2 N)$ amortized update complexity. 
Once an insertion/deletion causes the density of a segment to fall out of the range defined by $(\rho_i, \tau_i)$,
\pma tries to adjust the density by re-allocating all elements stored in the segment's parent. 
The adjustment process is invoked recursively and will only be terminated if all segments' densities fall
back into the range defined by \pma's density thresholds.
For an ordered array, modifications are trivial. Therefore, we mainly discuss insertions because deletions are the dual operation of insertions in \pma. 

\begin{example}\label{example:pma}
Figures \ref{fig:pma_insert_example} presents an example for \pma insertion. 
Each segment is uniquely identified by an interval (starting and ending position of the array) displayed in the corresponding tree node, 
e.g., the root segment is \textbf{segment-[0,31]} as it covers all 32 spaces. 
All values stored in \pma are displayed in the array.
The table in the figure shows predefined parameters including the segment size, 
the assignment of density thresholds ($\rho_i$, $\tau_i$) and the corresponding minimum and maximum entry sizes at different heights of the tree.
We use these setups as a running example throughout the paper.
To insert an entry. i.e. 48, into \pma, the corresponding leaf segment is firstly identified by a binary search, 
and the new entry is placed at the rear of leaf segment. 
The insertion causes the density of the leaf segment $4$ to exceed the threshold $3$. 
Thus, we need to identify the nearest ancestor segment which can accommodate the insertion without violating the thresholds, i.e., the \textbf{segment-[16,31]}.
Finally, the insertion is completed by re-distpatching all entries evenly in \textbf{segment-[16,31]}.
\end{example}

\marked{
\begin{lemma}[\cite{Bender:2005:COB,Bender:2007:PMA}]\label{lem:pma}
	The amortized update complexity of \pma is proved to be $O(\log^2 N)$ in
the worst case and $O(\log N)$ in the average case.
\end{lemma}
}

It is evident that \pma could be employed for dynamic graph maintenance 
as it maintains sorted elements efficiently with high locality on CPU.
However, the update procedure described in \cite{Bender:2007:PMA}
is inherently sequential and no concurrent algorithms have been proposed. 
To support batch updates of edge insertions and deletions for efficient graph stream analytic processing,
we devise \gpma to support concurrent \pma updates on GPUs.
\marked{
Note that we focus on the insertion process for a concise presentation because the deletion process is a dual process w.r.t. the insertion process in \pma.
}

\noindent\textbf{Concurrent Insertions in GPMA.}
Motivated by \pma on CPUs, we propose \gpma to handle a batch of insertions concurrently on GPUs. 
Intuitively, \gpma assigns an insertion to a thread and concurrently executes \pma algorithm for each thread with a lock-based approach to ensure consistency. 
More specifically, all leaf segments of insertions are identified in advance, 
and then each thread checks whether the inserted segments still satisfy their thresholds from bottom to top. 
For each particular segment, it is accessed in a mutually exclusive fashion. 
Moreover, all threads are synchronized after updating all segments located at the same tree height
to avoid possible conflicts as segments at a lower height are fully contained in the segments at a higher level. 

\begin{figure}[t]\centering
	\includegraphics[width=0.6\linewidth]{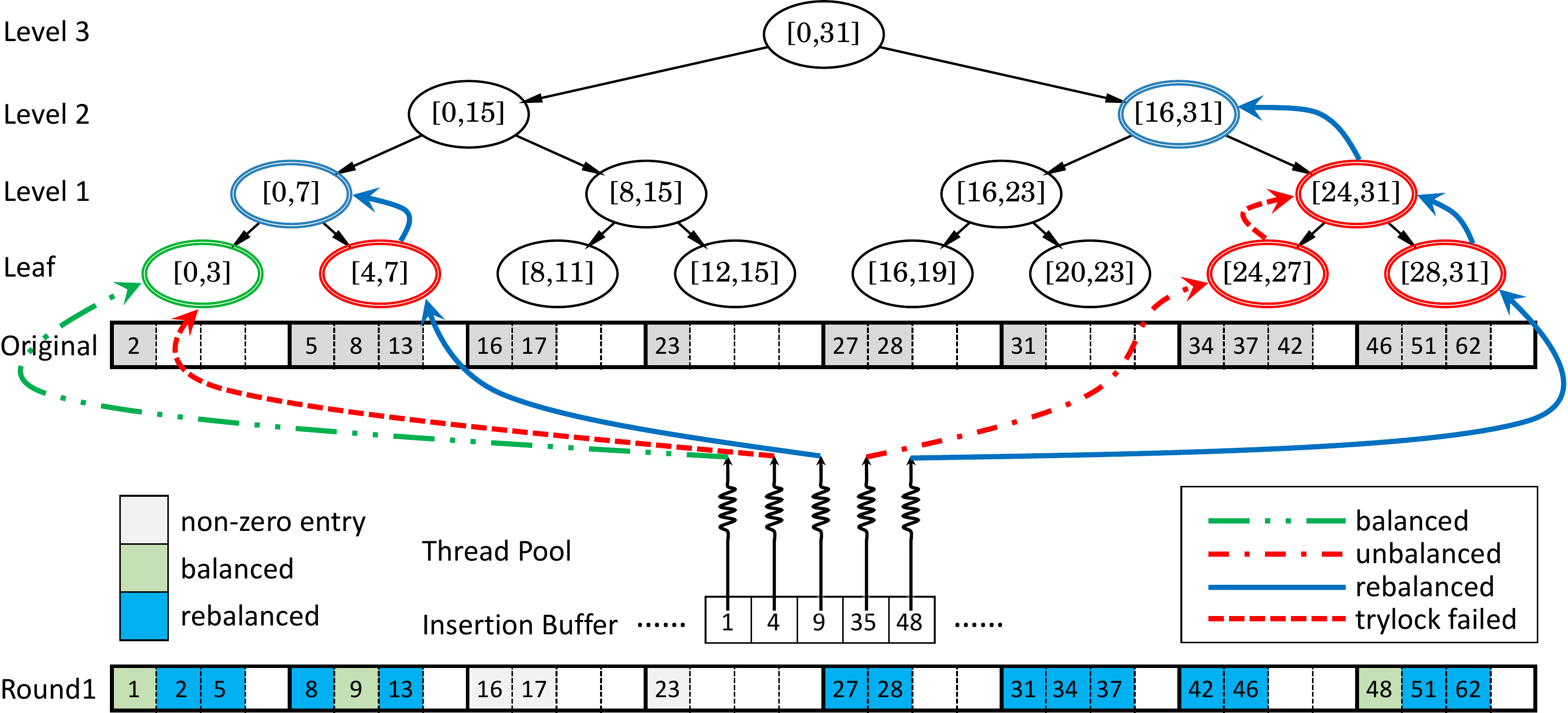}
	\caption{\marked{\gpma concurrent insertions}}
	\label{fig:naive_gpma_insert}
\end{figure}

Algorithm \ref{alg:pma_bulk_update} presents the pseudocode for \gpma concurrent insertions. 
We highlight the lines added to the original PMA update algorithm in order to achieve concurrent update of GPMA.
As shown in line \ref{alg:line:while_not_empty}, 
all entries in the insertion set are iteratively tried until all of them take effect. 
For each iteration shown in line \ref{alg:line:gpma1_try_insert}, all threads start at leaf segments and attempt the insertions in a bottom-up fashion. 
If a particular thread fails the mutex competition in line \ref{alg:line:trylock}, it aborts immediately and waits for the next attempt. 
Otherwise, it inspects the density of the current segment. 
If the current segment does not satisfy the density requirement, it will try the parent segment in the next loop iteration (lines \ref{alg:line:overflowed}-\ref{alg:line:overflow-end}). 
Once an ancestor segment is able to accommodate the insertion, it merges the new entry in line \ref{alg:line:merge_into_segment}
and the entry is removed from the insertion set. 
Subsequently, the updated segment will re-dispatch all its entries evenly and the process is terminated.

\begin{example}\label{emp:gpma_insert}
Figure \ref{fig:naive_gpma_insert} illustrates an example with five insertions, i.e. $\{1,4,9,35,48\}$, for concurrent \gpma insertion. 
The initial structure is the same as in Example~\ref{example:pma}. 
After identifying the leaf segment for insertion, 
threads responsible for \textbf{Insertion-1} and \textbf{Insertion-4} compete for the same leaf segment. 
Assuming \textbf{Insertion-1} succeeds in getting the mutex, \textbf{Insertion-4} is aborted.
Due to enough free space of the segment, \textbf{Insertion-1} is successfully inserted. 
Even though there is no leaf segment competition for \textbf{Insertions-9,35,48}, 
they should continue to inspect the corresponding parent segments because all the leaf segments do not satisfy the density requirement after the insertions. 
\textbf{Insertions-35,48} still compete for the same level-1 segment and \textbf{Insertion-48} wins. 
For this example, three of the insertions are successful and the results are shown in the bottom of Figure~\ref{fig:naive_gpma_insert}.
\textbf{Insertions-4,35} are aborted in this iteration and will wait for the next attempt.
\end{example}

\begin{figure}[t]
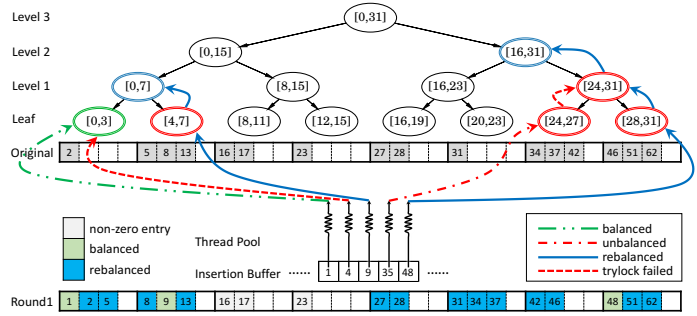

\begin{algorithm}[H]
\caption{GPMA Concurrent Insertion} 
\begin{algorithmic}[1]
	\Procedure{GPMAInsert}{Insertions $I$}
		\While{$I$ is \textbf{not} empty}\label{alg:line:while_not_empty}
			\State \textbf{\highlight{parallel} for} $i$ \textbf{in} $I$
				\Indent
					\State Seg $s$ $\gets$ \textsc{BinarySearchLeafSegment}(i)
					\State \textsc{TryInsert}(s, i, I)
				\EndIndent
			\State \textbf{\highlight{synchronize}}
			\State \highlight{release locks on all segments}
		\EndWhile
	\EndProcedure
	\Statex  
	\Function{TryInsert}{Seg $s$, Insertion $i$, Insertions $I$}
		\While{$s \neq $ root}\label{alg:line:gpma1_try_insert}
			\State \textbf{\highlight{synchronize}}
			\If{\highlight{fails to lock $s$}}\label{alg:line:trylock}
				\State \textbf{\highlight{return}} \Comment insertion aborts
			\EndIf

				\If{$(|s|+1)/\text{\emph{capacity}(s)}<\tau$}\label{alg:line:overflowed}
				\State $s$ $\gets$ parent segment of $s$  \label{alg:line:overflow-end}
				\Else
					\State \textsc{Merge}($s, i$)\label{alg:line:merge_into_segment}
					\State re-dispatch entries in $s$ evenly
					\State remove $i$ from $I$
					\State \textbf{return} \Comment insertion succeeds
				\EndIf

		\EndWhile
		\State double the space of the root segment
	\EndFunction
\end{algorithmic}
\label{alg:pma_bulk_update}
\end{algorithm}
\end{figure}

\marked{
\subsection{Adapting Graph Algorithms to {\Large \gpma}}\label{sec:generality}
Existing graph algorithms often use sparse matrix format to store the graph entries since most large graphs are naturally sparse\cite{Akoglu:2015:GBA}.
Although many different sparse storage formats have been proposed, most of the formats assume a specific order to organize the non-zero entries.
These formats enforce the order of the graph entries to optimize their specific access patterns,
e.g., row-oriented (COO\footnote{Generally, COO means ordered COO and it can also be column-oriented.}),  diagonal-oriented (JAD), and block-/tile-based (BCCOO, BRC and TCOO). 
It is natural that the ordered graph entries can be projected into an array and these similar formats can be supported by \gpma easily.
Among all formats, we choose CSR as an example to illustrate how to adapt the format to \gpma.

\vspace{1em}
\noindent\textbf{CSR as a case study.}
CSR is most widely used by existing algorithms on sparse matrices or graphs. 
CSR compresses COO's row indices into an offset array, which contributes to reducing the memory bandwidth when accessing the sparse matrix, and achieves a better workload estimation for skewed graph distribution (e.g., power-law distribution).
The following example demonstrates how to implement CSR on \gpma.
}
\begin{example}\label{emp:pmacsr}
In Figure~\ref{fig:csr_gpma}, we have a graph of three vertices and six edges.
The number on each edge denotes the weight of the corresponding edge. 
The graph is represented as a sparse matrix and is further transformed to the CSR format shown in the upper right. 
\marked{CSR sorts all non-zero entries in the row-orient order, and compresses row indices into intervals as a row offset array.}
The lower part denotes the \gpma representation of this graph. In order to maintain the row offset array without synchronization among threads, we add a guard entry whose column index is $\infty$ during concurrent insertions. That is to say, when the guard is moved, the corresponding element in row offset array will change.
\end{example}

\begin{figure}[t]
\begin{algorithm}[H]
\caption{Breadth-First Search} 
\begin{algorithmic}[1]
	\Procedure{BFS}{Graph \emph{G}, Vertex \emph{s}}
		\For{each vertex $u \in G.V - \{s\}$}
			\State $u.visited$ = \textbf{false}
		\EndFor
		\State $Q \gets \phi$
		\State $s.visited \gets$  \textbf{true}
		\State ENQUEUE($Q, s$)
		\While{$Q \neq \phi$}
			\State $u \gets$  DEQUEUE($Q$)
			\For{each $v \in G.Adj[u]$}
				\If{\highlight{IsEntryExist($v$)}}
					\If{$v.visited$ = \textbf{false}}
						\State $v.visited \gets$ \textbf{true}
						\State ENQUEUE($v$)
					\EndIf
				\EndIf
			\EndFor
		\EndWhile
	\EndProcedure
\end{algorithmic}
\label{alg:bfs}
\end{algorithm}
\end{figure}

\begin{figure}[t]
\begin{algorithm}[H]
\caption{GPU-based BFS Neighbour Gathering} 
\begin{algorithmic}[1]
	\Procedure{Gather}{Vertex \emph{frontier}, Int \emph{csrOffset}}
		\State \{\emph{r}, \emph{rEnd}\} $\gets$ \emph{csrOffset}[\emph{frontier}, \emph{frontier} + 1] \label{alg:line:get_offset}
		\For{($i$ $\gets$ $r$+threadId; $i$$<$\emph{rEnd}; $i$+=threadNum)}
			\If{\highlight{IsEntryExist($i$)}}\label{alg:line:pma_valid}
				ParallelGather($i$)
			\EndIf
		\EndFor
	\EndProcedure
\end{algorithmic}
\label{alg:bfs2}
\end{algorithm}
\end{figure}

\marked{
Given a graph stored on \gpma, the next step is to adapt existing graph algorithms
to \gpma. In particular, how existing algorithms access the graph entries stored on \gpma is of vital importance.
As for the CSR example, most algorithms access the entries by navigating through CSR's ordered array\cite{davidson2014work, kaleem2016synchronization, martone2010use, yang2011fast}.
We note that a CSR stored on \gpma is also an array which has bounded gaps interleaved with the graph entries.
Thus, we are able to efficiently replace the operations of arrays with the operations of \gpma. We will demonstrate how we can do this replacement as follows.


Algorithm~\ref{alg:bfs} illustrates the pseudocode of the classic BFS algorithm. 
We should pay attention to line 10, which is highlighted. 
Compared with the raw adjacency list, 
the applications based on \gpma need to guarantee the current vertex being traversed is a valid neighbour instead of an invalid space in \gpma's gap.

Algorithm~\ref{alg:bfs} provides a high-level view for \gpma adaption. Furthermore, we present how it adapts \gpma in the parallel GPU environment with some low-level details. Algorithm~\ref{alg:bfs2} is the pseudocode of
the \emph{Neighbour Gathering} parallel procedure, 
which is a general primitive for most GPU-based vertex-centric graph processing models \cite{Merrill:2015:HPS,davidson2014work,Fu:2014fy}. 
This primitive plays a role similar to line 10 of Algorithm~\ref{alg:bfs} but in a parallel fashion in accessing the neighbors of a particular vertex.
%
When traversing all neighbours of frontiers, \emph{Neighbour Gathering} follows the SIMT manner, 
which means that there are \emph{threadNum} threads as a group assigned to one of the vertex frontier 
and the procedure in Algorithm~\ref{alg:bfs2} is executed in parallel.
For the index range (in the CSR on \gpma) of the current frontier given by csrOffset (shown in line~\ref{alg:line:get_offset}), 
each thread will handle the corresponding tasks according to its \emph{threadId}. For GPU-based BFS, 
the visited labels of neighbours for all frontiers will not be judged immediately after neighbours are accessed. 
Instead, they will be compacted to contiguous memory in advance for higher memory efficiency.

Similarly, we can also carry out the entry existing checking for other graph applications to adapt them to \gpma. To summarize, \gpma can be adapted to common graph analytic applications which are
implemented in different representation and execution models, including matrix-based (e.g., PageRank), vertex-centric (e.g., BFS) and edge-centric (e.g., Connected Component). 
}

%

\begin{figure}[t]
\centering
\hspace{2mm}
\begin{minipage}{0.25\linewidth}
	\includegraphics[width=\linewidth]{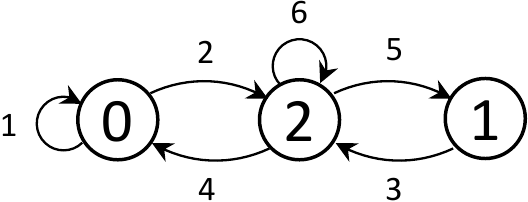}
	\vspace{-7mm}
	\caption*{\textbf{Example Graph}}
\end{minipage}
\hfill
\centering
\begin{minipage}{0.24\linewidth}
		\scriptsize
		\begin{tabular}[c]{rl}
			Row Offset &[0 2 3 6] \\
			Column Index &[0 2 2 0 1 2] \\
			Value &[1 2 3 4 5 6]
		\end{tabular}
		\vspace{2mm}
		\caption*{\textbf{CSR Format}}
\end{minipage}
\hfill
\centering
\begin{minipage}{0.4\linewidth}
	\includegraphics[width=\linewidth]{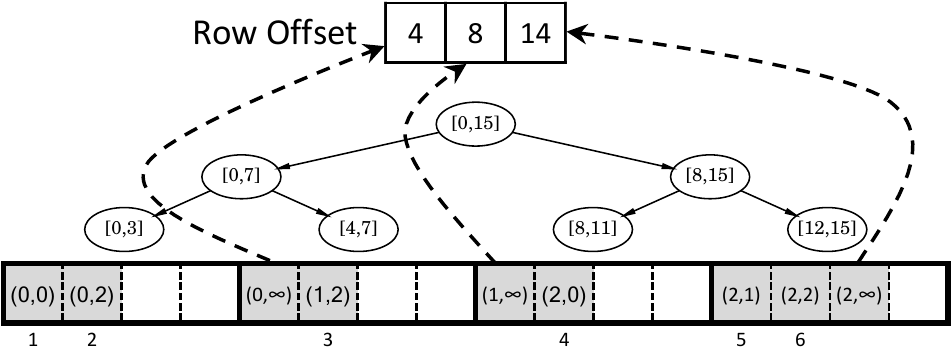}
\end{minipage}
\caption{\gpma based on CSR}
\label{fig:csr_gpma}
\end{figure}

\section{{\Large \gpmaplus}: {\Large \gpma} Optimization}\label{sec:gpmaplus}

Although \gpma can support concurrent graph updates on GPUs,
the update algorithm is basically a lock-based approach and can suffer 
from serious performance issue when different threads compete for the same lock. 
In this section, we propose a lock-free approach, i.e. \gpmaplus, 
which makes full utilization of GPU's massive multiprocessors.
We carefully examine the performance bottleneck of \gpma in Section~\ref{sec:bottle_analysis}. 
Based on the issues identified, we propose \gpmaplus for optimizing concurrent GPU updates with a lock-free approach in Section~\ref{sec:gpmaplus-update}. 

\subsection{Bottleneck Analysis}\label{sec:bottle_analysis}
The following four critical performance issues are identified for \gpma:

\begin{itemize}[noitemsep,leftmargin=*]
	\item \textbf{Uncoalesced Memory Accesses}: Each thread has to traverse the tree from the root segment to identify the corresponding leaf segment to be updated. 
	For a group of GPU threads which share the same memory controller (including access pipelines and caches), 
	memory accesses are uncoalesced and thus, cause additional IO overheads.

	\item \textbf{Atomic Operations for Acquiring Lock}: 
	Each thread needs to acquire the lock before it can perform the update. 
	Frequently invoking atomic operations for acquiring locks will bring huge overheads, especially for GPUs.


	\item \textbf{Possible Thread Conflicts}: When two threads conflict on a segment, one of them has to abort and wait for the next attempt. 
	In the case where the updates occur on segments which are located proximately,
	\gpma will end up with low parallelism. As most real world large graphs have the power law property, 
	the effect of thread conflicts can be exacerbated. 
	
	\item \textbf{Unpredictable Thread Workload}: 
	Workload balancing is another major concern for optimizing concurrent algorithms~\cite{stratton2012optimization}. 
	The workload for each thread in \gpma is unpredictable because: 
	(1) It is impossible to obtain the last non-leaf segment traversed by each thread in advance; 
	(2) The result of lock competition is random. 
	The unpredictable nature triggers the imbalanced workload issue for \gpma. 
	In addition, threads are grouped as warps on GPUs. 
	If a thread has a heavy workload, the remaining threads of the same warp are idle and cannot be re-scheduled. 
\end{itemize}

\begin{figure*}[t]\centering
	\includegraphics[width=\linewidth]{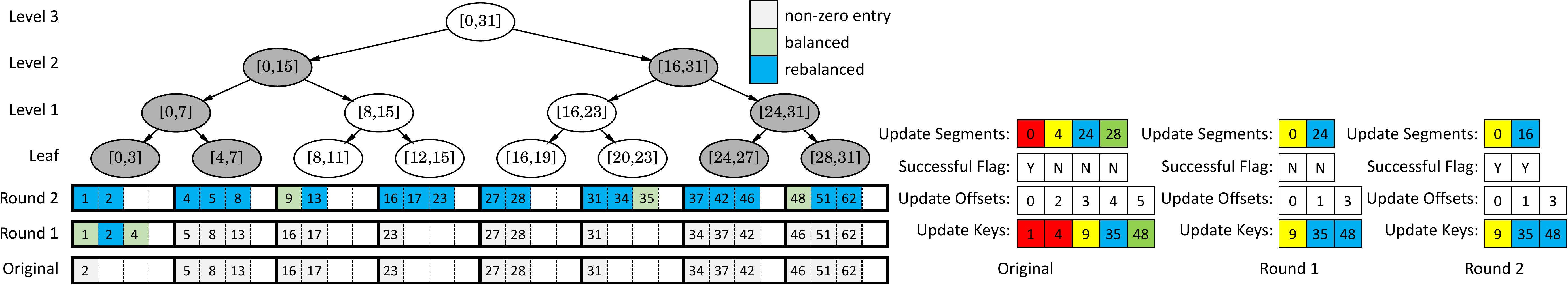}
	\caption{\gpmaplus concurrent insertions (best viewed in color)}
	\label{fig:gpma2_insert}
\end{figure*}

\subsection{Lock-Free Segment-Oriented Updates}\label{sec:gpmaplus-update}
Based on the discussion above, 
we propose \gpmaplus to lift all bottlenecks identified. 
The proposed \gpmaplus does not rely on lock mechanism and achieves high thread utilization simultaneously.
Existing graph algorithms can be adapted to GPMA+ in the same manner as GPMA.

Compared with \gpma, which handles each update separately, 
\gpmaplus concurrently processes updates based on the segments involved. 
It breaks the complex update pattern into existing concurrent GPU primitives to achieve maximum parallelism.
There are three major components in the \gpmaplus update algorithm:
%
\begin{enumerate}[label=(\arabic*),noitemsep,leftmargin=*]
\item The updates are first sorted by their keys and then dispatched to GPU threads for locating their corresponding leaf segments according to the sorted order. 
\item The updates belonging to the same leaf segment are grouped for processing and \gpmaplus processes the updates level by level in a bottom-up manner.
\item In any particular level, we leverage GPU primitives to invoke all computing resources for segment updates. 
\end{enumerate}
%
We note that, the issue of \emph{uncoalesced memory access} in \gpma is resolved by component (1)
as the updating threads are sorted in advance to achieve similar traversal paths.
Component (2) completely avoids the use of locks, which solves the problem of
\emph{atomic operations} and \emph{thread conflicts}.
Finally, component (3) makes use of GPU primitives to achieve \emph{workload balancing} among all GPU threads.



\begin{figure}[t]
\begin{algorithm}[H]
\caption{GPMA+ Segment-Oriented Insertion} 
\begin{algorithmic}[1]
	\Procedure{GpmaPlusInsertion}{Updates $U$}\label{alg:line:gpma_update}
		\State \textsc{Sort}($U$)\label{alg:line:sort_by_key}
		\State Segs $S \gets \textsc{BinarySearchLeafSegments}$($U$)\label{alg:line:binary_search}
		\While{root segment is not reached}\label{alg:line:do_begin}
		\State Indices $I \gets \emptyset$
		\State Segs $S^* \gets \emptyset$
		\State $(S^*,I)$ $\gets$ \textsc{UniqueSegments}($S$)\label{alg:line:unique_insertion}
		\State \textbf{parallel for} $s \in S^*$ \label{alg:line:while_statement}
		\Indent
			\State \textsc{TryInsert+}($s$, $I$, $U$)\label{alg:line:try_insert}
		\EndIndent 
		\If{$U = \emptyset$}
			\State \textbf{return}\label{alg:line:return_successful}
		\EndIf
		\State \textbf{parallel for} $s \in S$ \label{alg:line:parent_begin}
		\Indent
			\If{$s$ does not contain any update}
				\State remove $s$ from $S$
			\EndIf
			\State $s \gets $ parent segment of $s$\label{alg:line:parent_end}
		\EndIndent
	\EndWhile\label{alg:line:do_end}
	\State $r \gets$ double the space of the old root segment
	\State \textsc{TryInsert+}($r$, $\emptyset$, $U$)
	\EndProcedure
	\Statex
	\Function{UniqueSegments}{Segs $S$}
		\State $(S^*, Counts) \gets \textsc{RunLengthEncoding}(S)$\label{alg:line:rle}
		\State Indices $I \gets \textsc{ExclusiveScan}(Counts)$\label{alg:line:exscan}
		\State \Return $(S^*,I)$
	\EndFunction
	\Statex
	\Function{TryInsert+}{Seg $s$, Indices $I$, Updates $U$}
		\State $n_s$ $\gets$ \textsc{CountSegment}($s$)\label{alg:line:campact_valid_elements}
		\State $U_s$ $\gets$ \textsc{CountUpdatesInSegment}($s$,$I$,$U$)
		\If{$(n_s + |U_s|)/capacity(s)<\tau$}
			\State \textsc{Merge}($s$, $U_s$)\label{alg:gpma_plus_merge}
			\State re-dispatch entries in $s$ evenly
			\State remove $U_s$ from $U$
		\EndIf\label{alg:line:end_try_insert}
	\EndFunction
\end{algorithmic}
\label{alg:pma2_bulk_update}
\end{algorithm}
\end{figure}

We present the pseudocode for \gpmaplus's segment-oriented insertion in the procedure \newline \textsc{GpmaPlusInsertion} of Algorithm \ref{alg:pma2_bulk_update}. 
Note that, similar to Section~\ref{sec:gpma} (\gpma), we focus on presenting the insertions for \gpmaplus and the deletions could be naturally inferred. 
The inserting entries are first sorted by their keys in line~\ref{alg:line:sort_by_key} and the corresponding segments are then identified in line~\ref{alg:line:binary_search}.
Given the update set $U$, 
\gpmaplus processes updating segments level by level in lines~\ref{alg:line:do_begin}-\ref{alg:line:do_end} until all updates are executed successfully (line~\ref{alg:line:return_successful}).
In each iteration, \textsc{UniqueInsertion} in line~\ref{alg:line:unique_insertion} groups update entries belonging to the same segments into unique segments, i.e., $S^*$,
and produces the corresponding index set $I$ for quick accesses of updates entries located in a segment from $S^*$. 
As shown in lines~\ref{alg:line:rle}-\ref{alg:line:exscan}, \textsc{UniqueSegments} only utilizes standard GPU primitives, 
i.e. \textsc{RunLenghtEncoding} and \textsc{ExclusiveScan}. \textsc{RunLenghtEncoding} compresses an input array by merging runs of an element into a single element. 
It also outputs a count array denoting the length of each run.
\textsc{ExclusiveScan} calculates, for each entry $e$ in an array, the sum of all entries before $e$.
Both primitives have very efficient parallelized GPU-based implementation which makes full utilization of the massive GPU cores. 
In our implementation, we use the NVIDIA CUB library \cite{CUB} for these primitives.
Given a set of unique updating segments, \textsc{TryInsert+} first checks if a segment $s$ has enough space for accommodating the updates 
by summing the valid entries in $s$ (\textsc{CountSegment}) and the number of updates in $s$ (\textsc{CountUpdatesInSegment}). 
If the density threshold is satisfied, the updates will be materialized by merging the inserting entries with existing entries in the segment (as shown in line~\ref{alg:gpma_plus_merge}).
Subsequently, all entries in the segment will be re-dispatched to balance the densities.    
After \textsc{TryInsert+}, the algorithm will terminate if there are no entries to be updated. 
Otherwise, \gpmaplus will advance to higher levels by setting all remaining segments to their parent segments (lines~\ref{alg:line:parent_begin}-\ref{alg:line:parent_end}).
The following example illustrates \gpmaplus's segment-oriented updates.

\marked{
\begin{example}\label{emp:gpmaplus}
Figure \ref{fig:gpma2_insert} illustrates an example for \gpmaplus insertions with the same setup as in example \ref{emp:gpma_insert}. 
The left part is \gpmaplus's snapshots in different rounds during this batch of insertions. The right part denotes the corresponding array information after the execution of each round. 
Five insertions are grouped into four corresponding leaf segments (denoted in different colours). 

For the first iteration at the leaf level, \textbf{Insertions-1,4} of the first segment (denoted as red) is merged into the corresponding leaf segment, then its success flag is marked and will not be considered in the next round. The remaining intervals fail in this iteration and their corresponding segments will upgrade to their parent segments. It should be noted that the blue and the green grids belong to the same parent segment and therefore, will be merged and then dispatched to their shared parent segment (as shown in Round1). In this round, both segments (denoted as yellow and blue) cannot satisfy the density threshold, and their successful flags are not checked. In Round2, both update segments can be merged by the corresponding insertions and no update segments will be considered in the next round since all of them are flagged.
\end{example}
}

In Algorithm \ref{alg:pma2_bulk_update}, 
\textsc{TryInsert+} is the most important function as it handles all the corresponding insertions with no conflicts.
Moreover, it achieves a balanced workload for each concurrent task.
This is because \gpmaplus handles the updates level by level and each segment to be updated in a particular level has
exactly the same capacity.
However, segments in different levels have different capacities.
Intuitively, the probability of updating a segment with a larger size (a segment closer to the root) is much lower than
that of a segment with a smaller size (a segment closer to the leaf). 
To optimize towards the GPU architecture, we propose the following optimization strategies for \textsc{TryInsert+}
for segments with different sizes.



\begin{itemize}[noitemsep,leftmargin=*]
	\item \textbf{Warp-Based}: For a segment with entries not larger than the warp size, the segment will be handled by a warp. 
	Since all threads in the same warp are tied together and warp-based data is held by registers,
	updating a segment by a warp does not require explicit synchronization and will obtain superior efficiency.
	
	\item \textbf{Block-Based}: For a segment of which the data can be loaded in GPU's shared memory, block-based approach is chosen. 
	Block-based approach executes all updates in the shared memory. 
	As shared memory has much larger size than warp registers, block-based approach can handle large segments efficiently. 

	\item \textbf{Device-Based}: For a segment with the size larger than the size of the shared memory, we handle them via global memory and rely on kernel synchronization.  
	Device-based approach is slower than the two approaches above, but it has much less restriction on memory size (less than device memory amount) and is not invoked frequently.
\end{itemize}

We refer interested readers to Appendix~\ref{sec:cuda_implementation} for the detailed algorithm of the optimizations above.

\begin{theorem}\label{them:complexity}
Given there are $K$ computation units in the GPU, the amortized update performance of \gpmaplus is $O(1+\frac{log^2N}{K})$,
where $N$ is the maximum number of edges in the dynamic graph. 
\end{theorem}
\begin{proof}
Let $X$ denote the set of updating entries contained in a batch. 
We consider the case where $|X| \geq K$ as it is rare to see $|X| < K$ in real world scenarios. 
In fact, our analysis works for cases where $|X| = O(K)$. 
The total update complexity consists of three parts: (1) sorting the updating entries;
(2) searching the position of the entries in \gpma; (3) inserting the entries.
We study these three parts separately.

For part (1), the sorting complexity of $|X|$ entries on the GPU is $O(\frac{|X|}{K})$
since parallel radix sort is used (keys in \gpma are integers for storing edges). Then, the amortized sorting complexity is $O(\frac{|X|}{K})/|X|=O(1)$. 

For part (2), the complexity of concurrently searching $|X|$ entries on \gpma is $O(\frac{|X|\cdot logN}{K})$ since each
entry is assigned to one thread and the depth of traversal is the same for one thread (\gpma is a balanced tree). 
Thus, the amortized searching complexity is $O(\frac{|X|\cdot logN}{K})/|X|=O(\frac{logN}{K})$. 

For part (3), we need to conduct a slightly complicated analysis. 
We denote the total insertion complexity of $X$ with \gpmaplus as $c^X_{\gpmaplus}$.
As \gpmaplus is updated level by level, 
$c^X_{\gpmaplus}$ can be decomposed into:
$c^X_{\gpmaplus} = c_0 + c_1 + ... + c_h$
where $h$ is the height of the \pma tree. 
\marked{
Given any level $i$, let $z_i$ denote the number of segments to be updated by \gpmaplus.
}
Since all segments at level $i$ have the same size, we denote $p_i$ as the sequential complexity to update any segment $s_{i,j}$ at level $i$ (\textsc{TryInsert+} in Algorithm~\ref{alg:pma2_bulk_update}).
\gpmaplus evenly distributes the computing resources to each segment.
As processing each segment only requires a constant number of scans on the segment by GPU primitives, 
the complexity for \gpmaplus to process level $i$ is $c_i = \frac{p_i \cdot z_i}{K}$. 
Thus we have:
\begin{align*}
c^X_{\gpmaplus}& = \sum_{i = 0,..,h}\frac{p_i \cdot z_i}{K} \leq \frac{1}{K}\sum_{x \in X}c^x_{\pma}
\end{align*}
where $c^x_{\pma}$ is the sequential complexity for \pma to process the update of a particular entry $x \in X$.
\marked{
The inequality holds because for each segment updated by \gpmaplus, it must be updated at least once by a sequential \pma process.
With Lemma~\ref{lem:pma}, we have $c^x_{\pma} = O(log^2 N)$ and thus
$c^X_{\gpmaplus} = O(\frac{|X|\cdot log^2 N}{K})$.
} 
Then the amortized complexity to update one
single entry under the \gpma scheme naturally follows as $O(1+\frac{log^2 N}{K})$.

Finally, we conclude the proof by combining the complexities from all three parts.
\end{proof}
\marked{

Theorem~\ref{them:complexity} proves that the speedups of \gpmaplus over sequential \pma
is linear to the number of processing units available on GPUs, which showcases the theoretical scalability of \gpmaplus.
}

\section{Experimental Evaluation}\label{sec:experiment}
In this section, we present the experimental evaluation
of our proposed methods. First, we present the setup of the experiments.
Second, we examine the update costs of different schemes for
maintaining dynamic graphs.
Finally, we implement
three different applications to show the performance and the scalability of the proposed solutions.

\begin{table*}
\centering
\caption{Experimented Graph Algorithms and the Compared Approaches}
\scriptsize
\label{tab:approaches}
\begin{tabular}{|c|c|c|c|c|}
\hline
Compared Approaches & Graph Container & \apbfs & \apcc & \appagerank \\ \hline
\multirow{3}{*}{CPU Approaches} & \adjlists & \multicolumn{3}{|c|}{\multirow{2}{*}{Standard Single Thread Algorithms}} \\ \cline{2-2}
& \pma \cite{Bender:2005:COB, Bender:2007:PMA} & \multicolumn{3}{|c|}{}  \\ \cline{2-5}
& \stinger \cite{Ediger:2012ve} & \multicolumn{3}{c|}{\stinger built-in Parallel Algorithms} \\ \hline
\multirow{2}{*}{GPU Approaches} & \gpucsr \cite{cuSparse} & \multirow{2}{*}{D. Merrill et al.\cite{Merrill:2015:HPS}} & \multirow{2}{*}{J. Soman et al.\cite{Soman:2010vm}} & \multirow{2}{*}{CUSP SpMV \cite{cusp}} \\ \cline{2-2}
& \gpma/\gpmaplus & & & \\ \hline
\end{tabular}
\end{table*}

\begin{figure*}[t]\centering
	\includegraphics[width=\linewidth]{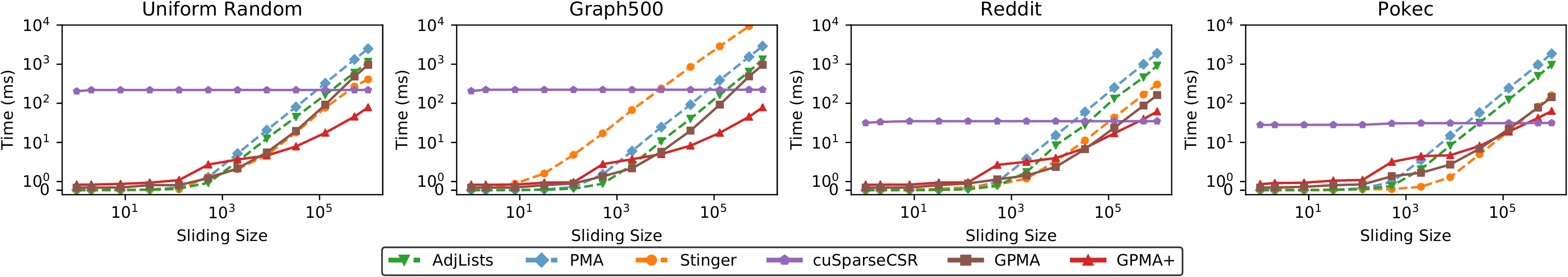}
	\caption{Performance comparison for updates with different batch sizes. The dashed lines represent CPU-based solutions whereas the solid lines represent GPU-based solutions.}
	\label{fig:pma_insert_compare}
\end{figure*}

\subsection{Experimental Setup}\label{sec:experiment:setup}

\noindent\textbf{Datasets.}
We collect two real world graphs (\dsreddit and \dspokec)
and synthesize two random graphs (\dsrandom and \dsgraphlib) to test the proposed methods.
The datasets are described as follows and their statistics are
summarized in Table \ref{tab:datasets}.
\begin{itemize}[noitemsep,leftmargin=*]
\item \dsreddit is an online forum where user actions
include post and comment. We collect all comment actions
from a public resource\footnote{https://www.kaggle.com/reddit/reddit-comments-may-2015}.
Each comment of a user $b$ to a post from another user $a$
is associated with an edge from $a$ to $b$, and the edge indicates an action of $a$ has triggered
an action of $b$. As each comment is labeled with a timestamp, it naturally
forms a dynamic influence graph.
\item \dspokec is the most popular online social network in Slovakia.
We retrieve the dataset from SNAP \cite{Leskovec:2016:SNAP}.
Unlike other online datasets, \dspokec contains the whole
network over a span of more than 10 years. Each edge corresponds to a friendship
between two users.
\item \dsgraphlib is a synthetic dataset obtained by using the Graph500 RMAT
generator \cite{Murphy:2010:ITG} to synthesize a large power law graph.
\item \dsrandom is a random graph generated by the Erd\H{o}s-Renyi model.
Specifically, given a graph with $n$ vertices, the random graph is generated by
including each edge with probability $p$. In our experiments, we generate a
Erd\H{o}s-Renyi random graph with $0.02\%$ of non-zero entries against a full clique.
\end{itemize}


\begin{table}[t]\centering
\caption{Statistics of Datasets}
\label{tab:datasets}
\begin{tabular}{|c|c|c|c|c|c|}
\hline
Datasets & ~~$|\vset|$~~ & $~~|\eset|~~$ & $|\eset|/|\vset|$ & $|\eset_s|$ & $|\eset_s|/|\vset|$\\ \hline
\dsreddit	& 2.61M	& 34.4M		& 13.2	& 17.2M	 & 6.6	\\ \hline
\dspokec 	& 1.60M	& 30.6M		& 19.1	& 15.3M  & 9.6 	\\ \hline
\dsgraphlib & 1.00M & 200M		& 200	& 100M	& 100	\\ \hline
\dsrandom 	& 1.00M	& 200M		& 200	& 100M	& 100	\\ \hline
\end{tabular}
\end{table}


\noindent\textbf{Stream Setup.}
In our datasets, \dsreddit has a timestamp on every edge whereas
the other datasets do not possess timestamps. As commonly used in existing graph stream algorithms \cite{zhang2016approximate,yang2016tracking,ohsaka2015efficient},
we randomly set the timestamps of all edges in the \dspokec, \dsgraphlib and
\dsrandom datasets.
Then, the graph stream of each dataset receives the edges with increasing timestamps.

For each dataset, a dynamic graph stream is initialized
with a subgraph consisting of the dataset's first half of its total edges according to the timestamps,
i.e., $\eset_s$ in Table~\ref{tab:datasets} denotes
the initial edge set of a dynamic graph before the stream starts.
To demonstrate the update performance of both insertions and deletions,
we adopt a sliding window setup where the window contains a fixed number of edges.
Whenever the window slides, we need to update the graph
by deleting expired edges and inserting arrived edges
until there are no new edges left in the stream.

\vspace{1mm}
\noindent\textbf{Applications.}
We conduct experiments on three most widely used graph applications to showcase the applicability and the efficiency of \gpmaplus.

\begin{itemize}[noitemsep,leftmargin=*]
	\item \textbf{BFS} is a key graph operation which is extensively studied in previous works on GPU graph processing \cite{harish2007accelerating, luo2010effective, busato2015bfs}.
It begins with a given vertex (or \emph{root})
of an unweighted graph and iteratively explores all connected vertices.
The algorithm will assign a minimum distance away from the root vertex to every visited vertex after
it terminates. In the streaming scenario, after each graph update,
we select a random root vertex and perform \apbfs from the root to explore the entire graph.
\item \textbf{Connected Component} is another fundamental algorithm which has been extensively studied 
under both CPU \cite{hirschberg1976parallel} and GPU \cite{Soman:2010vm}  µ environment.
It partitions the graph in the way that all vertices in a partition can reach the others in the same partition
and cannot reach vertices from other partitions. In the streaming context, 
after each graph update, we run the \apcc algorithm to maintain the up-to-date partitions.
\item \textbf{PageRank} is another popular benchmarking application for large scale graph processing.
Power iteration method is a standard method to evaluate the PageRank where the Sparse Matrix Vector
Multiplication (\spmv) kernel is recursively executed between the graph's adjacency matrix and the PageRank vector. In the streaming scenario, whenever
the graph is updated, the power iteration is invoked and it obtains the up-to-date PageRank vector
by operating on the updated graph adjacency matrix
and the PageRank vector obtained in the previous iteration.
In our experiments, we follow the standard setup by setting the damping factor to 0.85 and
we terminate the power iteration once the 1-norm error is less than $10^{-3}$.
\end{itemize}

These three applications have different memory and computation requirements. BFS requires little computation but performs frequent random memory accesses, 
and PageRank using \spmv accesses the memory sequentially and it is the most compute-intensive task among all three applications.

\vspace{1mm}
\noindent\textbf{Maintaining Dynamic Graph.}
We adopt the CSR \cite{Liu:2016:ICB,Ashari:2014:FSM} format to represent the dynamic graph
maintained. Note that all approaches proposed in the paper
are not restricted to CSR but general enough to incorporate
any popular representation formats like COO \cite{dang2012sliced}, JAD \cite{saad1989numerical}, HYB \cite{Bell:2008:ESM, martone2010use} and many others.
To evaluate the update performance of our proposed methods, we compare different graph data structures and respective approaches
on both CPUs and GPUs. 
\begin{itemize}[noitemsep,leftmargin=*]
\item \adjlists(CPU). \adjlists is a basic approach for CSR graph representation. As the CSR format sorts 
all entries according to their row-column indices, we implement \adjlists with a vector of $|V|$ entries for $|V|$ vertices
and each entry is a RB-Tree to denote all (out)neighbors of each vertex. 
The insertions/deletions are operated by TreeSet insertions/deletions. 
\item \pma (CPU). We implement the original CPU-based \pma and adopt it for the CSR format. 
The insertions/deletions are operated by \pma insertions/deletions. 
\item \stinger (CPU). We compare the graph container structure
used in the state-of-the-art CPU-based parallel dynamic graph analytic system, \stinger \cite{Ediger:2012ve}.
The updates are handled by the internal logic of \stinger.
\item \gpucsr (GPU). We also compare with the GPU-based CSR format used in the NVIDIA cuSparse library~\cite{cuSparse}. 
The updates are executed by calling the rebuild function in the cuSparse library. 
\item \gpma/\gpmaplus. \marked{They are our proposed approaches.
Although insertions and deletions could be handled similarly,
in the sliding window models where the numbers of insertions and deletions are often equal,
the lazy deletions can be performed via marking the location as deleted without triggering the density maintenance and recycling for new insertions.
}
\end{itemize}
%
Note that we do not compare with DCSR \cite{King:2016:DSM} because, as discussed in Section~\ref{sec:literature:analytics},
the scheme can neither handle deletions nor support efficient
searches, which makes it incomparable to all schemes proposed in this paper.

To validate if using the dynamic graph format proposed in this paper affects the performance of graph algorithms,
we implement the state-of-the-art GPU-based algorithms on the CSR format maintained by
\gpma/\gpmaplus as well as \gpucsr.
Meanwhile, we invoke \stinger's built-in APIs to handle the same workloads of the graph algorithms,
which are considered as the counterpart of GPU-based approaches in highly parallel CPU environment.
Finally, we implement the standard single-threaded algorithms for each application in \adjlists and \pma as baselines for thorough evaluation.
The details of all compared solutions for each application is summarized in Table~\ref{tab:approaches}.

\vspace{1mm}
\noindent\textbf{Experimental Environment.} All algorithms mentioned in the remaining part
of this section are implemented
with CUDA 7.5 and GCC 4.8.4 with -O3 optimization.
All experiments except \stinger run on a CentOS server which has Intel(R) Core i7-5820k (6-cores, 3.30GHz) with 64GB main memory and three GeForce TITAN X GPUs (each has 12GB device memory), connected with PCIe v3.0.
\stinger baselines run on a multi-core server which is deployed 4-way Intel(R) Xeon(R) CPU E7-4820 v3 (40-cores, 1.90GHz) with 128GB main memory.

\subsection{The Performance of Handling Updates}

\begin{figure*}[t]

\begin{minipage}{\linewidth}
	\includegraphics[width=\linewidth]{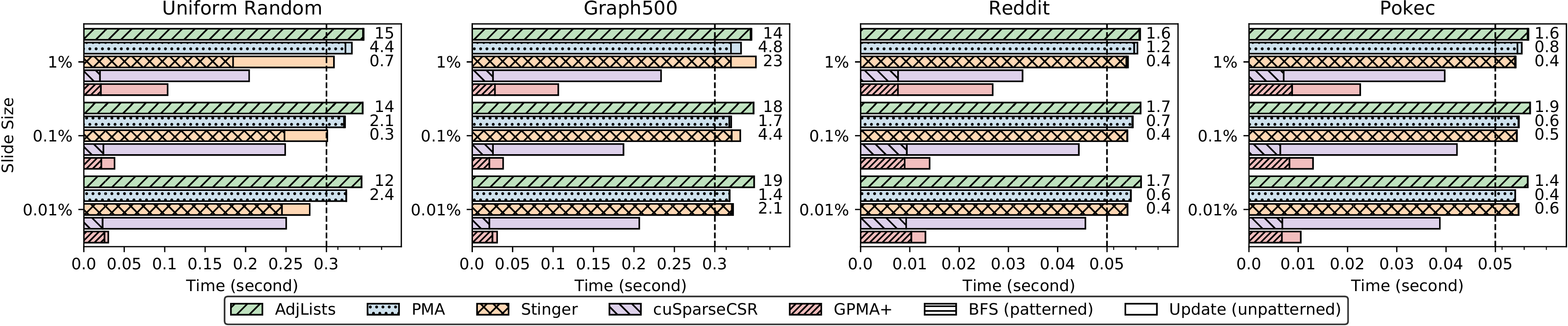}
	\caption{Streaming BFS}
	\label{fig:result_bfs}
\end{minipage}

\begin{minipage}{\linewidth}
	\includegraphics[width=\linewidth]{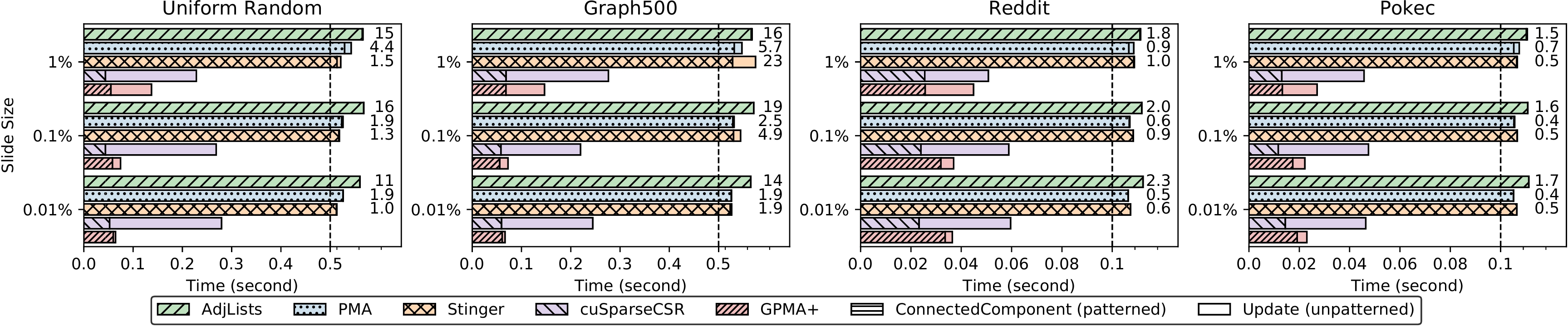}
	\caption{Streaming Connected Component}
	\label{fig:result_cc}
\end{minipage}

\begin{minipage}{\linewidth}
	\includegraphics[width=\linewidth]{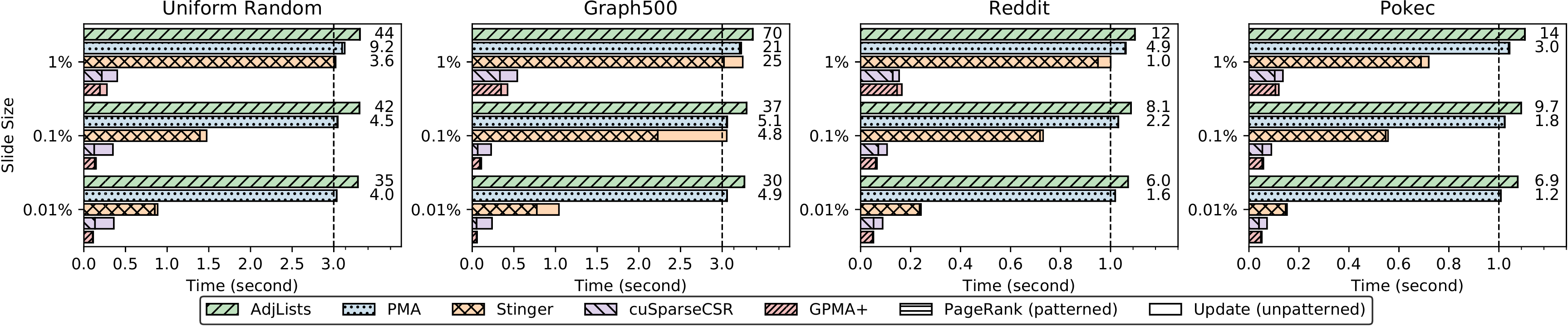}
	\caption{Streaming PageRank}
	\label{fig:result_pagerank}
\end{minipage}

\end{figure*}

\begin{figure*}[t]

\begin{minipage}{\linewidth}
	\includegraphics[width=\linewidth]{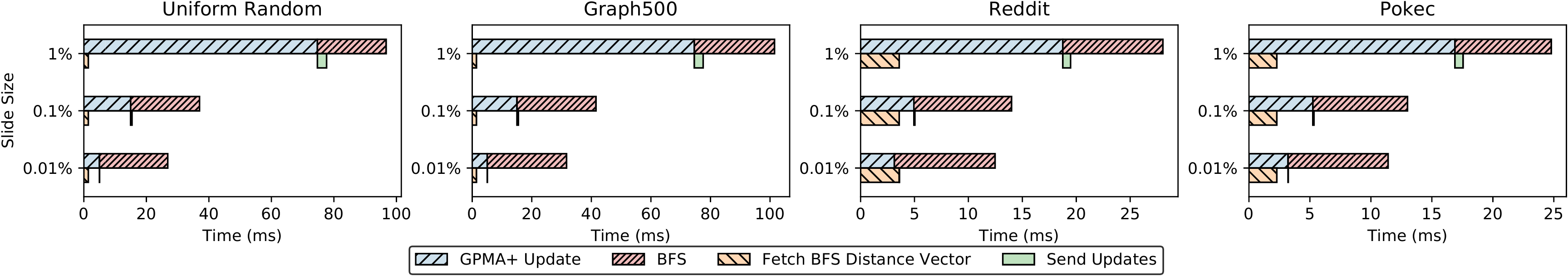}
	\caption{Concurrent data transfer and BFS computation with asynchronous stream}
	\label{fig:result_bfs_mem}
\end{minipage}

\begin{minipage}{\linewidth}
	\includegraphics[width=\linewidth]{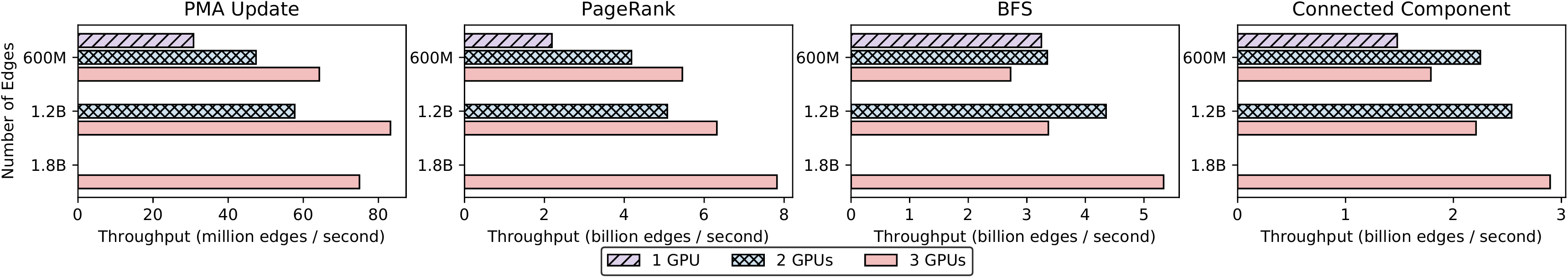}
	\caption{Multi-GPU performance on different sizes of \dsgraphlib datasets}
	\label{fig:multi-gpu}
\end{minipage}
\end{figure*}

In this subsection, we compare the update costs for different update approaches.
As previously mentioned, we start with the initial subgraph consisting of each dataset's first half of total edges.
We measure the average update time where the sliding window iteratively shifts for a batch of edges. 
To evaluate the impact of update batch sizes, the batch size is set to range from one edge and exponentially grow to one million edges with base two. 
Figure~\ref{fig:pma_insert_compare} shows the average latency for all approaches with different sliding batch sizes. 
Note that the x-axis and y-axis are plotted in log scales.
\marked{We have also tested sorted graph streams to evaluate extreme cases. We omit the detailed results and we refer interested readers to Appendix~\ref{sec:oreder_graphs_result}.}

We observe that, \pma-based approaches are very efficient in handling updates when the batch size is small.
As batch size becomes larger, the performance of \pma and \gpma quickly degrades to the performance of simple rebuild.
Although \gpma achieves better performance than \gpmaplus for small batches since the concurrent updating entries are unlikely to conflict,
thread conflicts become serious for larger batches.
Due to its lock-free characterstic, \gpmaplus shows superior performance over \pma and \gpma.
In particular, \gpmaplus has speedups of up to 20.42x and 18.30x against \pma and \gpma respectively.
\stinger shows impressive update performance in most cases as \stinger efficiently updates its dynamic graph structure in a parallel fashion and the code runs on a powerful multi-core CPU system.
For now, multi-core CPU system is considered more powerful than GPUs for pure random data structure maintainance but cost more (in our experimental setup, our CPU server costs more than 5 times that of the GPU server).
Moreover, we also note that, \stinger shows extremely poor performance in the \dsgraphlib dataset. 
According to the previous study \cite{bader2009stinger}, the phenomenon is due to the fact that \stinger holds a fixed size of each edge block.
Since \dsgraphlib is a heavily skewed graph as the graph follows the power law model,
the skewness causes severe performance deficiency on the utilization of memory for \stinger.


We observe the sharp increase for \gpmaplus performance curves occur when the batch size is $512$.
This is because the multi-level strategy is used in \gpmaplus (which is mentioned in Section~\ref{sec:gpmaplus-update}) 
and shared-memory constraint cannot support batch size which is more than 512 on our hardware.
Finally, the experiments show that, \gpma is faster than \gpmaplus when the update batch is smaller and leads to few thread conflicts, 
because the \gpmaplus logic is more complicated and includes overheads by a number of kernel calls. 
However, using \gpma only benefits when the update batch is extremely small and the performance gain in
such extreme case is also negligible compared with \gpmaplus. 
Hence, we can conclude that \gpmaplus shows its stability and efficiency across different update patterns compared with \gpma, 
and we will only show the results of \gpmaplus in the remaining experiments.

\subsection{Application Performance}\label{sec:experiment:app}
As previously mentioned, all compared application-specific approaches are summarized in Table~\ref{tab:approaches}.
We find that integrating \gpmaplus into an existing GPU-based implementation requires little modification. The main one is in transforming the array operations in the original implementation to the operations on \gpmaplus, as presented in Section~\ref{sec:generality}.
The intentions of this subsection are two-fold.
First, we test if using the \pma-like data structure to represent the graph brings significant overheads for the graph algorithms.  
Second, we demonstrate how the update performance affects the overall efficiency of dynamic graph processing.


In the remaining part of this section, we present the performance of different approaches by showing
their average elapsed time to process a shift of the sliding window with three different batch sizes, 
i.e., the batches contain $0.01\%$, $0.1\%$ and $1\%$ edges of the respective dataset.
We have also tested the graph stream with explicit random insertions and deletions for all applications as an extended experiment. We omit the detailed results here since they are similar to the results of the sliding window model and we refer interested readers to Appendix~\ref{sec:explicit_delete_expr}. 
We distinguish the time spent on
updates and analytics with different patterns among all figures.

\vspace{1mm}\noindent


\noindent\textbf{\emph{BFS Results:}}
Figure~\ref{fig:result_bfs} presents the results for \apbfs.
Although processing \apbfs only accesses each edge in the graph once,
it is still an expensive operation because \apbfs can potentially scan the entire graph.
This has led to the observation that CPU-based approach takes significant amount of time for BFS computation
whereas the update time is comparatively negligible.
Thanks to the massive parallelism and high memory bandwidth of GPUs, GPU-based approaches are much more efficient
than CPU-based approaches for BFS computation as well as the overall performance. 
For the \gpucsr approach, the rebuild process is the bottleneck as the update needs to scan the entire group multiple times. 
In contrast, \gpmaplus takes much shorter time for the update and has nearly identical BFS performance compared with \gpucsr. 
Thus, \gpmaplus dominates the comparisons in terms of the overall processing efficiency. 



We have also tested our framework in terms of hiding data transfer over PCIe by using asynchronous streams to concurrently
perform GPU computation and PCIe transfer. In Figure~\ref{fig:result_bfs_mem}, we show the results
when running concurrent execution by using the \gpmaplus approach.
The data transfer consists of two parts: sending graph updates and fetching updated distance vector (from the query vertex to all other vertices).
It is clear from the figure that, under any circumstances,
sending graph updates is overlapped by \gpmaplus update processing and
fetching the distance vector is overlapped by \apbfs computation.
Thus, the data transfer is completely hidden in the concurrent streaming scenario.
As the observations remain similar in other applications, we omit their results and explanations, and the details can be found in Appendix~\ref{sec:async_data_transfer}.

\vspace{1mm}\noindent


\noindent\textbf{\emph{Connected Component Results:}}
Figure~\ref{fig:result_cc} presents the results for running {\tt Connected Component} on the dynamic graphs.
The results show different performance patterns compared with \apbfs as
\apcc takes more time in processing which is caused by a number of graph traversal passes to extract the partitions.  
Meanwhile, the update cost remains the same. 
Thus, GPU-based solutions enhance their performance superiority over CPU-based solutions. 
Nevertheless, the update process of \gpucsr is still expensive compared with the time spent on
\texttt{Connected- Component}. 
\gpmaplus is very efficient in processing the updates. Although we have observed that,
in the \dsreddit and the \dspokec datasets, \gpmaplus shows some discrepancies for
running the graph algorithm against \gpucsr due to the ``holes'' introduced in the graph structure, 
the discrepancies are insignificant considering the huge performance boosts for updates. 
Thus, \gpmaplus still dominates the rebuild approach for overall performance. 

\vspace{1mm}\noindent


\noindent\textbf{\emph{PageRank Results:}}
Figure~\ref{fig:result_pagerank} presents the results for \texttt{Page- Rank}.
\appagerank is a compute-intensive task where the \spmv kernel is iteratively invoked on the entire graph until the PageRank vector converges.
The pattern follows from previous results:
CPU-based solutions are dominated by GPU-based approaches because
iterative \spmv is a more expensive process than \apbfs and \apcc, and
GPU is designed to handle massively parallel computation like \spmv.
Although \gpucsr shows inferior performance compared with \gpmaplus, 
the improvement brought by \gpmaplus's efficient update is not as significant as that in previous applications since
the update costs are small compared with the cost of iterative \spmv kernel calls. 
Nevertheless, the dynamic structure of \gpmaplus does not affect the efficiency of the \spmv kernel
and \gpmaplus outperforms other approaches in all experiments.

\marked{
\subsection{Scalability}
\gpma and \gpmaplus can also be extended to multiple GPUs to support graphs with size larger than the device memory of one GPU.
To showcase the scalability of our proposed framework, we implement the multi-GPU version of \gpmaplus and carry out experiments of the aforementioned graph applications.

We generate three large datasets using \dsgraphlib with increasing numbers of edges (600 Million, 1.2 Billion and 1.8 Billion) and 
conduct the same performance experiments in section~\ref{sec:experiment:app} with 1\% slide size, on 1, 2 and 3 GPUs respectively. 
We evenly partition graphs according to the vertex index and synchronize all devices after each iteration. 
For fair comparison among different datasets, we use throughput as our performance metric.
The experimental results of \gpmaplus updates and application performance are illustrated in Figure~\ref{fig:multi-gpu}.
We do not compare with \stinger because in this subsection, we focus on the evaluation on the scalability of \gpmaplus. The memory consumption of \stinger exceeds our machine's 128GB main memory based on its default configuration in the standalone mode. 

Multiple GPUs can extend the memory capacity so that 
analytics on larger graphs can be executed. 
According to Figure~\ref{fig:multi-gpu}, 
the improvement in terms of throughput for multiple GPUs behaves differently in various applications. 
For \gpmaplus update and \appagerank, we achieve a significant improvement with more GPUs, because their workloads between communications are relatively compute-intensive.
For \apbfs and \apcc, the experimental results demonstrate a tradeoff between overall computing power and communication cost with increasing number of GPUs, 
as these two applications incur larger communication cost. 
Nevertheless, multi-GPU graph processing is an emerging research area and more effectiveness optimizations are left as future work.
Overall, this set of preliminary experiments shows that our proposed scheme is capable of supporting large scale dynamic graph analytics.
}
\subsection{Overall Findings}
We summarize our findings in the subsection. First, GPU-based approaches (\gpucsr and \gpmaplus)
outperform CPU-based approaches thanks to our optimizations in taking advantage of the superior hardware of the GPUs, even compared with \stinger running on a 40-core CPU server.
One of the key reasons is that \gpmaplus and graph analytics can exploit the superb high memory bandwidth and massive parallelism of the GPU, as many graph applications are data- and compute-intensive.
Second, \gpmaplus is much more efficient than \gpucsr as maintaining the dynamic updates is often the bottleneck of continuous graph analytic processing
and \gpmaplus avoids the costly process of rebuilding the graph via incremental updates while bringing minimal overheads for existing graph algorithms running its
graph structure.


\section{Conclusion \& Future Work}\label{sec:conclusion}

In this paper, we address how to dynamically update the graph structure on GPUs in an efficient manner. 
First, we introduce a GPU dynamic graph analytic framework, 
which enables existing static GPU-oriented graph algorithms to support high-performance evolving graph analytics. 
Second, to avoid the rebuild of the graph structure which is a bottleneck for processing dynamic graphs on GPUs, 
we propose \gpma and \gpmaplus to support incremental dynamic graph maintenance in parallel.
We prove the scalability and complexity of \gpmaplus theoretically and evaluate the efficiency through extensive experiments. 
\marked{As the future work, we would like to explore the hybrid CPU-GPU supports for dynamic graph processing and more effectiveness optimizations for involved applications.}

\clearpage
\clearpage
\appendix
\section*{Appendices}
\addcontentsline{toc}{section}{Appendices}
\renewcommand{\thesubsection}{\Alph{subsection}}

\subsection{\textsc{TryInsert+} Optimizations}\label{sec:cuda_implementation}

Based on different segment sizes, we propose three optimizations of Function \textsc{TryInsert+} in Algorithm~\ref{alg:pma2_bulk_update}. 
The motivation is to obtain better memory access efficiency and lower cost of synchronization by balancing between problem scale and hardware hierarchy on GPU.
The key computation logic of \textsc{TryInsert+} is to merge two sorted arrays, i.e., existing segment entries and entries to be inserted. 
Standard approach for parallel merging needs to identify the position in merged array by binary search and then to execute \emph{parallel map},
which requires heavy and uncoalesced memory accesses.
Thus, depending on the size of the merge,
we wish to employ different hardware hierarchies on GPU (i.e. warp, block and device) to minimize the cost of memory accesses.

Before presenting the details of our optimizations, 
Algorithm~\ref{alg:try_insert_init} illustrates how to group threads according to their positions in different hierarchies of GPU architecture
and how to target the groups to their assigned segments.
In particular, each thread is assigned with a lane id, a block id and a global thread id 
to indicate the position of the thread in the corresponding warp, block and device work group.
Each thread is assigned for one \gpmaplus segment and the thread will ask other threads in the same work group to cooperate for its task. 
This means that each thread tries to drive a group of threads to deal with the assigned segment.
Such a strategy lifts thread divergences caused by different execution branches.
Note that this assignment policy will be used in our warp and block based optimizations as an initialization function. 

Algorithm~\ref{alg:try_insert_warp} shows the Warp-Based optimization for any segments with entries no larger than the warp size. 
This implementation has high efficiency because explicit synchronization is not needed and all data is stored in registers. 
For each iteration, all threads of a particular warp will compete for the control of the warp as shown in line 11. 
The winner will drive the other threads in this warp to handle its required computation steps of the corresponding segment.
As an example, line 27 counts valid entries in the segment concurrently. 
Lines 32-34 omit the remaining computation steps in \textsc{TryInsert+}, such as merging insertions and redistributing entries of segments, 
as their computation paradigm is similar to counting entries.

Algorithm~\ref{alg:try_insert_block} shows the Block-Based optimization. 
It utilizes the shared memory, which has a higher volume than registers, to store data. 
Even though explicit synchronization is needed in line 12 and line 32 to guarantee consistency,
synchronization in a block is highly optimized in GPU hardware and thus it does little effect to the overall performance.
Both Warp-Based and Block-Based optimizations explicitly accommodate GPU features. 
As discussed in Section~\ref{sec:gpmaplus-update}, although these two methods have limited memory for efficient access, 
they can handle most of the update requests.

Algorithm~\ref{alg:try_insert_device} shows the Device-Based implementation. 
The implementation is different from the ones in Warp and Block based approaches, 
because it is designed for segments having a size larger than the shared memory size. 
Under this scenario, we have to handle them in the GPU's global memory. 
One possible approach is to invoke an independent kernel for each large segment, 
but it will take considerable costs to initialize and schedule for multiple kernels. 
Hence, we propose an approach to handle a large number segments by only invoking 
a few unique kernel calls.

We illustrate the idea by showing how to perform counting segments which are valid for insertions as an example. 
As shown in lines 5 and 12, all valid entries stored in \gpmaplus segments are first marked, and then all valid entry counts are calculated by \emph{SumReduce} in one kernel call. Line 16 generates valid indexes for segments which have enough free space to receive their corresponding insertions, which is used by the rest computation steps. 
Simply speaking, our approach executes in horizontal steps of the execution logic, in order to avoid load imbalance caused by branch divergences.
Finally, merging and segment entries redistribution use the same mechanism and thus are omitted.

\clearpage

\begin{algorithm}[t]
\begin{lstlisting}[basicstyle=\scriptsize\ttfamily, language=c++, numbers=left]   
inline function ConstInit( void ) {
  // cuda protocol variables
  WARPS = blockDim / 32;
  warp_id = threadIdx / 32;
  lane_id = threadIdx % 32;
  thread_id = threadIdx;
  block_width = gridDim;
  grid_width = gridDim * blockDim;
  global_id = block_width * blockIdx + threadIdx;
  
  // infos for assigned segment
  seg_beg = segments[global_id];
  seg_end = seg_beg + segment_width;
  
  // infos for insertions belong current segment
  ins_beg = offset[global_id];
  ins_end = offset[global_id + 1];
  insert_size = ins_end - ins_beg;
  
  // the upper number that current segment can hold
  upper_size = tau * segment_size;
}
\end{lstlisting}
\caption{\textsc{TryInsert+} Initialization}
\label{alg:try_insert_init}
\end{algorithm}

\begin{algorithm}[b] 
\begin{lstlisting}[basicstyle=\scriptsize\ttfamily, language=c++, numbers=left]         
kernel TryInsert+(int segments[m], int offsets[m],
    int insertions[n], int segment_width) {
  
  ConstInit();

  volatile shared comm[WARPS][5];
  warp_shared_register pma_buf[32]; 
  
  while (WarpAny(seg_end - seg_beg)) {
    // compete for warp	
    comm[warp_id][0] = lane_id;
    
    // winner controls warp in this iteration
    if (comm[warp_id][0] == lane_id) {
      seg_beg = seg_end;
      comm[warp_id][1] = seg_beg;
      comm[warp_id][2] = seg_end;
      comm[warp_id][3] = ins_beg;
      comm[warp_id][4] = ins_end;
    }
    
    memcpy(pma_buf, pma[seg_beg], segment_width);
    // count valid entries in this segment
    entry_num = 0;
    if (lane_id < segment_width) {
      valid = pma_buf[lane_id] == NULL ? 0 : 1;
      entry_num = WarpReduce(valid);
    }
    
    // check upper density if insert
    if (entry_num + insert_size) < upper_size) {
      // merge insertions with pma_buf
      // evenly redistribute pma_buf
      // mark all insertions successful
      memcpy(pma[seg_beg], pma_buf, segment_width);
    }
  }  
}	
\end{lstlisting}
\caption{\textsc{TryInsert+} Warp-Based Optimization}
\label{alg:try_insert_warp}
\end{algorithm}

\begin{algorithm}[t]  
\begin{lstlisting}[basicstyle=\scriptsize\ttfamily, language=c++, numbers=left]             
kernel TryInsert+(int segments[m], int offsets[m],
    int insertions[n], int segment_width) {
	
  ConstInit();

  volatile shared comm[5];
  volatile shared pma_buf[segment_width]; 
  
  while (BlockAny(seg_end - seg_beg)) {
    // compete for block	
    comm[0] = thread_id;
    BlockSynchronize();
    
    // winner controls block in this iteration
    if (comm[0] == lane_id) {
      seg_beg = seg_end;
      comm[1] = seg_beg;
      comm[2] = seg_end;
      comm[3] = ins_beg;
      comm[4] = ins_end;
    }
    
    memcpy(pma_buf, pma[seg_beg], segment_width);
    // count valid entries in this segment
    entry_num = 0;
    ptr = thread_id;
    while (ptr < segment_width) {
      valid = pma_buf[ptr] == NULL ? 0 : 1;
      entry_num += BlockReduce(valid);
      thread_id += block_width;
    }
    BlockSynchronize();
    
    // same as lines 30-37 in Algorithm 5
  }  
}
\end{lstlisting}
\caption{\textsc{TryInsert+} Block-Based Optimization} 
\label{alg:try_insert_block}
\end{algorithm}

\begin{algorithm}[b]
\begin{lstlisting}[basicstyle=\scriptsize\ttfamily, language=c++, numbers=left]                    
function TryInsert+(int segments[m], int offsets[m],
    int insertions[n], int segment_width) {
    
  int valid_flags[m * segment_width];
  parallel foreach i in range(m):
    parallel foreach j in range(segment_width):
      if (pma[segments[i] + j] != NULL) {
        valid_flags[i * segment_width + j] = 1;
      }   
  DeviceSynchronize();
  int entry_nums[m];
  DeviceSegmentedReduce(valid_flags, m,
      segment_size, entry_nums);
  DeviceSynchronize();
  int valid_indexes[m];
  parallel foreach i in range(m):
    if (entry_nums[i] + insert_size < upper_size) {
      valid_indexes[i] = i;
    }
  DeviceSynchronize();
  RemoveIfTrue(valid_indexes);
  DeviceSynchronize();
  // according to valid_indexes, segmentedly to:
  //   merge insertions into segments
  //   evenly redistribute segments
  //   mark all insertions successful 
}
\end{lstlisting}
\caption{\textsc{TryInsert+} Device-Based Optimization} 
\label{alg:try_insert_device}
\end{algorithm}

\clearpage
\subsection{Additional Experimental Results For Data Transfer}\label{sec:async_data_transfer}
We show the experimental results for using asynchronous streams
for concurrently transmitting data on PCIe and running computations on the GPU.
We only show the results for \gpmaplus.

In \apcc, the data transferred on PCIe consists of two parts: the graph updates 
and the component label vector to all vertices computed by \apcc. In \appagerank, the result vector to be fetched indicates PageRank scores, which has the same size as \apcc's. 
The results in Figures~\ref{fig:result_cc_mem} and \ref{fig:result_pagerank_mem}
have shown that the data transfer is completely hidden by analytic processing on GPU
and \gpmaplus update. 

\begin{figure}[ht]
\centering
\includegraphics[width=\linewidth]{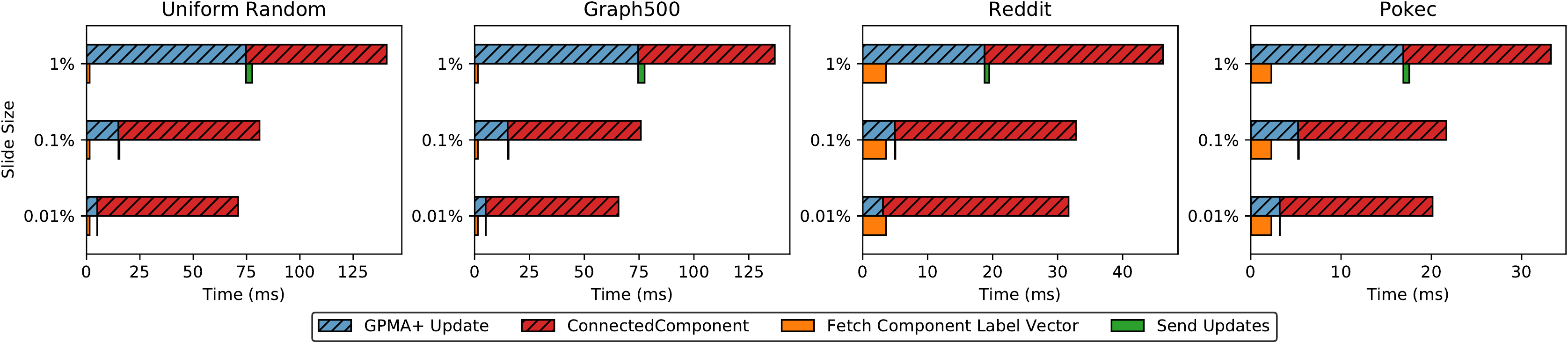}
\caption{Concurrent data transfer and Connected Component computation with asynchronous stream}
\label{fig:result_cc_mem}
\end{figure}

\begin{figure}[ht]
\centering
\includegraphics[width=\linewidth]{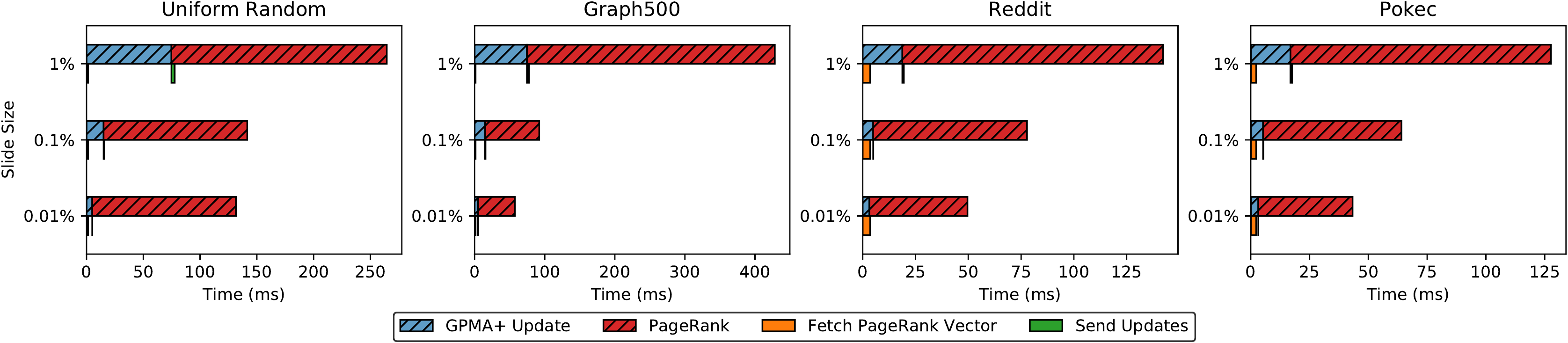}
\caption{Concurrent data transfer and PageRank computation with asynchronous stream}
\label{fig:result_pagerank_mem}
\end{figure}

\subsection{The Performance of Handling Updates on Sorted Graphs}\label{sec:oreder_graphs_result}
\begin{figure*}[h]\centering
	\includegraphics[width=\linewidth]{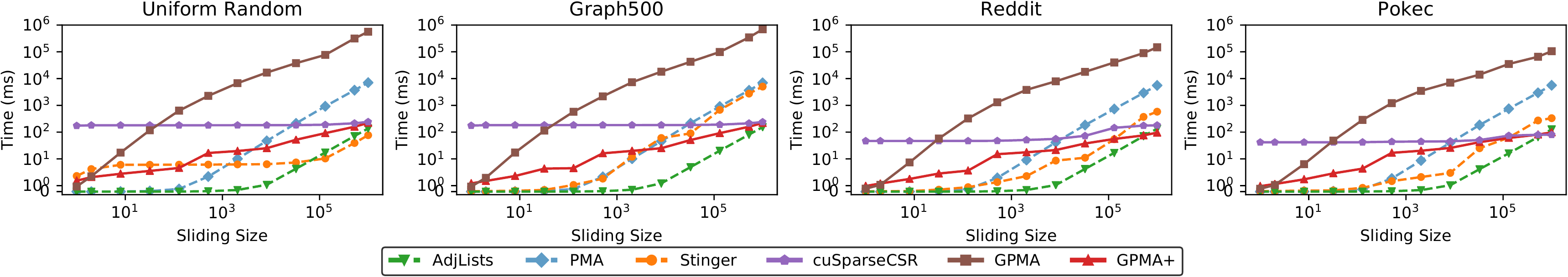}
	\caption{Performance comparison for updates with different batch sizes. The dashed lines represent CPU-based solutions whereas the solid lines represent GPU-based solutions.}
	\label{fig:pma_insert_compare_pos}
\end{figure*}

For the update results with sorted streaming orders,
\adjlists performs the best among all approaches due to its efficient balanced binary tree update mechanism. 
Meanwhile, a batch of sorted updates makes \gpma very inefficient as all updating threads within the batch conflict. 
Thanks to the non-locking optimization introduced, the update performance of \gpmaplus is still significantly faster than that of the rebuild approach (\gpucsr)
with orders of magnitude speedups for small batch sizes.

\clearpage
\subsection{Additional Experimental Results For Graph Streams with Explicit Deletions}\label{sec:explicit_delete_expr}
We present the experimental results for graph streams which involve explicit deletions. In this section, we use the same stream setup which is mentioned in Section~\ref{sec:experiment:setup}. However, for each iteration of sliding, we will randomly pick a set of edges belonging to current sliding window insteading of the head part as edges to be deleted.

\begin{figure}[ht]
\centering
\includegraphics[width=\linewidth]{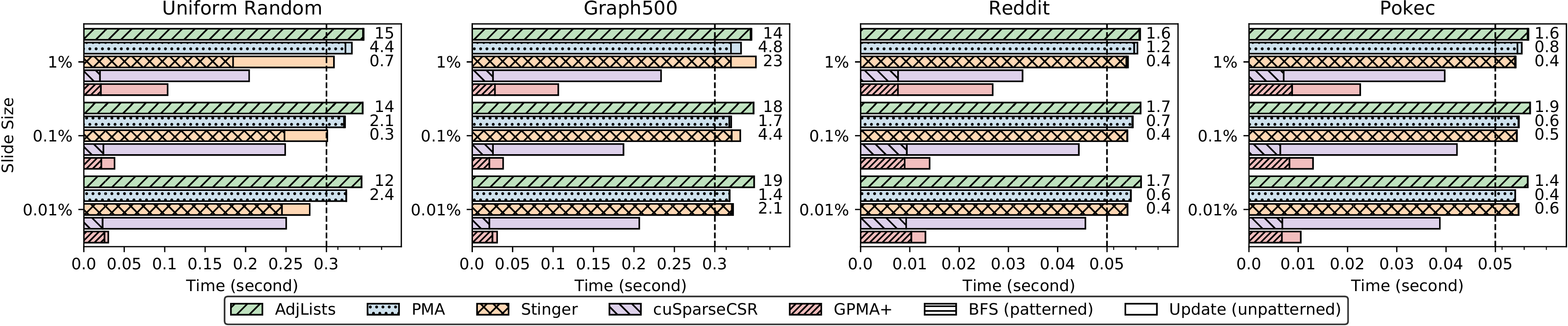}
\caption{Streaming BFS with explicit deletions}
\label{fig:result_bfs_r}
\end{figure}

\begin{figure}[ht]
\centering
\includegraphics[width=\linewidth]{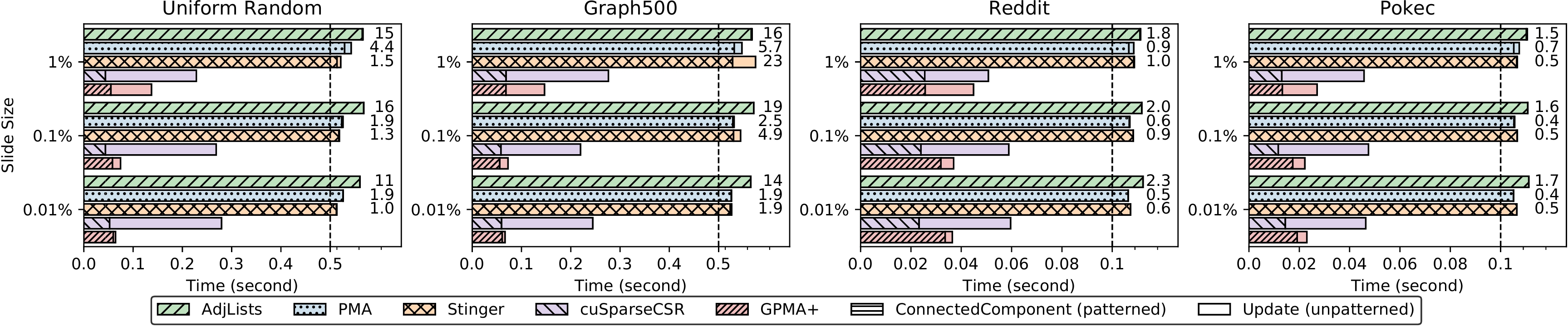}
\caption{Streaming Connected Component with explicit deletions}
\label{fig:result_cc_r}
\end{figure}

\begin{figure}[ht]
\centering
\includegraphics[width=\linewidth]{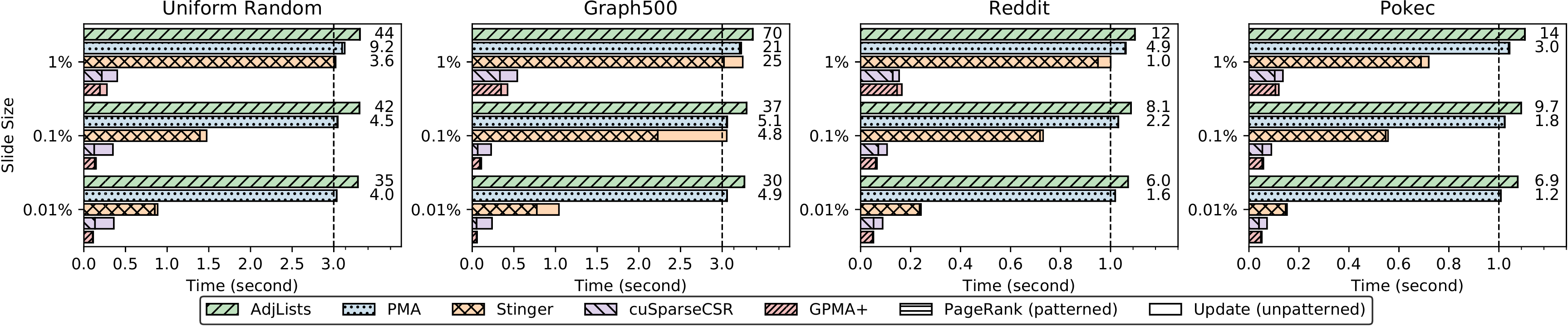}
\caption{Streaming PageRank with explicit deletions}
\label{fig:result_pagerank_r}
\end{figure}

Figures~\ref{fig:result_bfs_r}, \ref{fig:result_cc_r} and \ref{fig:result_pagerank_r} illustrate the results of three streaming applications respectively. Note that we pick sets of edges to be deleted in advance, which means that for each independent baseline, it handles the same workload all the time. Since there is no intrinsic difference between expiry and explicit deletions, the results are similar to sliding window's.
The subtle difference in the results are mainly due to different deletions which lead to various applications' running time.

\clearpage
\bibliographystyle{abbrv}
\bibliography{reference}

\begin{thebibliography}{10}

\bibitem{Flink}
Apache flink.
\newblock \url{https://flink.apache.org/}.
\newblock Accessed: 2016-10-18.

\bibitem{cusp}
Cusp library.
\newblock \url{https://developer.nvidia.com/cusp}.
\newblock Accessed: 2017-03-25.

\bibitem{cuSparse}
cusparse.
\newblock \url{https://developer.nvidia.com/cusparse}.
\newblock Accessed: 2016-11-09.

\bibitem{CUB}
{CUDA} {UnBound} ({CUB}) library.
\newblock \url{https://nvlabs.github.io/cub/}, 2015.

\bibitem{Akoglu:2015:GBA}
L.~Akoglu, H.~Tong, and D.~Koutra.
\newblock Graph based anomaly detection and description: a survey.
\newblock {\em Data Min. Knowl. Discov.}, 29(3):626--688, 2015.

\bibitem{Ashari:2014:FSM}
A.~Ashari, N.~Sedaghati, J.~Eisenlohr, S.~Parthasarathy, and P.~Sadayappan.
\newblock Fast sparse matrix-vector multiplication on gpus for graph
  applications.
\newblock In {\em SC}, pages 781--792, 2014.

\bibitem{ashari2014efficient}
A.~Ashari, N.~Sedaghati, J.~Eisenlohr, and P.~Sadayappan.
\newblock An efficient two-dimensional blocking strategy for sparse
  matrix-vector multiplication on gpus.
\newblock In {\em ICS}, pages 273--282, 2014.

\bibitem{bader2009stinger}
D.~A. Bader, J.~Berry, A.~Amos-Binks, D.~Chavarr{\'\i}a-Miranda, C.~Hastings,
  K.~Madduri, and S.~C. Poulos.
\newblock Stinger: Spatio-temporal interaction networks and graphs (sting)
  extensible representation.
\newblock {\em Georgia Institute of Technology, Tech. Rep}, 2009.

\bibitem{Bell:2008:ESM}
N.~Bell and M.~Garland.
\newblock Efficient sparse matrix-vector multiplication on {CUDA}.
\newblock Technical Report NVR-2008-004, NVIDIA Corporation, 2008.

\bibitem{Bender:2005:COB}
M.~A. Bender, E.~D. Demaine, and M.~Farach-Colton.
\newblock Cache-oblivious b-trees.
\newblock {\em SIAM J. Comput.}, 35(2):341--358, 2005.

\bibitem{Bender:2007:PMA}
M.~A. Bender and H.~Hu.
\newblock An adaptive packed-memory array.
\newblock {\em {ACM} Trans. Database Syst.}, 32(4), 2007.

\bibitem{Braun:2015:AMH}
L.~Braun, T.~Etter, G.~Gasparis, M.~Kaufmann, D.~Kossmann, D.~Widmer,
  A.~Avitzur, A.~Iliopoulos, E.~Levy, and N.~Liang.
\newblock Analytics in motion: High performance event-processing and real-time
  analytics in the same database.
\newblock In {\em SIGMOD}, pages 251--264, 2015.

\bibitem{busato2015bfs}
F.~Busato and N.~Bombieri.
\newblock Bfs-4k: an efficient implementation of bfs for kepler gpu
  architectures.
\newblock {\em TPDS}, 26(7):1826--1838, 2015.

\bibitem{Cheng:2012:KTP}
R.~Cheng, J.~Hong, A.~Kyrola, Y.~Miao, X.~Weng, M.~Wu, F.~Yang, L.~Zhou,
  F.~Zhao, and E.~Chen.
\newblock Kineograph: Taking the pulse of a fast-changing and connected world.
\newblock In {\em EuroSys}, pages 85--98, 2012.

\bibitem{crouch2013dynamic}
M.~S. Crouch, A.~McGregor, and D.~Stubbs.
\newblock Dynamic graphs in the sliding-window model.
\newblock In {\em European Symposium on Algorithms}, pages 337--348. Springer,
  2013.

\bibitem{dang2012sliced}
H.-V. Dang and B.~Schmidt.
\newblock The sliced coo format for sparse matrix-vector multiplication on
  cuda-enabled gpus.
\newblock {\em Procedia Computer Science}, 9:57--66, 2012.

\bibitem{datar2002maintaining}
M.~Datar, A.~Gionis, P.~Indyk, and R.~Motwani.
\newblock Maintaining stream statistics over sliding windows.
\newblock {\em SIAM journal on computing}, 31(6):1794--1813, 2002.

\bibitem{davidson2014work}
A.~Davidson, S.~Baxter, M.~Garland, and J.~D. Owens.
\newblock Work-efficient parallel gpu methods for single-source shortest paths.
\newblock In {\em Parallel and Distributed Processing Symposium, 2014 IEEE 28th
  International}, pages 349--359. IEEE, 2014.

\bibitem{Ediger:2012ve}
D.~Ediger, R.~McColl, E.~J. Riedy, and D.~A. Bader.
\newblock {STINGER - High performance data structure for streaming graphs.}
\newblock {\em HPEC}, 2012.

\bibitem{lkin:2011:SFD}
M.~Elkin.
\newblock Streaming and fully dynamic centralized algorithms for constructing
  and maintaining sparse spanners.
\newblock {\em ACM Trans. Algorithms}, 7(2):20:1--20:17, 2011.

\bibitem{Feigenbaum:2005:OGP}
J.~Feigenbaum, S.~Kannan, A.~McGregor, S.~Suri, and J.~Zhang.
\newblock On graph problems in a semi-streaming model.
\newblock {\em Theor. Comput. Sci.}, 348(2-3):207--216, 2005.

\bibitem{Fu:2014fy}
Z.~Fu, M.~Personick, and B.~Thompson.
\newblock {\em MapGraph: A High Level API for Fast Development of High
  Performance Graph Analytics on GPUs}.
\newblock A High Level API for Fast Development of High Performance Graph
  Analytics on GPUs. ACM, New York, New York, USA, June 2014.

\bibitem{Guha:2012:GSS}
S.~Guha and A.~McGregor.
\newblock Graph synopses, sketches, and streams: {A} survey.
\newblock {\em {Proc. VLDB Endow.}}, 5(12):2030--2031, 2012.

\bibitem{harish2007accelerating}
P.~Harish and P.~Narayanan.
\newblock Accelerating large graph algorithms on the gpu using cuda.
\newblock In {\em International Conference on High-Performance Computing},
  pages 197--208. Springer, 2007.

\bibitem{hirschberg1976parallel}
D.~S. Hirschberg.
\newblock Parallel algorithms for the transitive closure and the connected
  component problems.
\newblock In {\em Proceedings of the eighth annual ACM symposium on Theory of
  computing}, pages 55--57. ACM, 1976.

\bibitem{Iyer:2016:TGP}
A.~P. Iyer, L.~E. Li, T.~Das, and I.~Stoica.
\newblock Time-evolving graph processing at scale.
\newblock In {\em Proceedings of the Fourth International Workshop on Graph
  Data Management Experiences and Systems}, pages 5:1--5:6, 2016.

\bibitem{Iyer:2015:RTC}
A.~P. Iyer, L.~E. Li, and I.~Stoica.
\newblock Celliq : Real-time cellular network analytics at scale.
\newblock In {\em NSDI}, pages 309--322, 2015.

\bibitem{kaleem2016synchronization}
R.~Kaleem, A.~Venkat, S.~Pai, M.~Hall, and K.~Pingali.
\newblock Synchronization trade-offs in gpu implementations of graph
  algorithms.
\newblock In {\em Parallel and Distributed Processing Symposium, 2016 IEEE
  International}, pages 514--523. IEEE, 2016.

\bibitem{King:2016:DSM}
J.~King, T.~Gilray, R.~M. Kirby, and M.~Might.
\newblock Dynamic sparse-matrix allocation on gpus.
\newblock In {\em ISC}, pages 61--80, 2016.

\bibitem{Leskovec:2016:SNAP}
J.~Leskovec and R.~Sosi{\v{c}}.
\newblock Snap: A general-purpose network analysis and graph-mining library.
\newblock {\em TIST}, 8(1):1, 2016.

\bibitem{Lin:2014:ESM}
X.~Lin, R.~Zhang, Z.~Wen, H.~Wang, and J.~Qi.
\newblock Efficient subgraph matching using gpus.
\newblock In {\em ADC}, pages 74--85, 2014.

\bibitem{Liu:2016:ICB}
H.~Liu, H.~H. Huang, and Y.~Hu.
\newblock ibfs: Concurrent breadth-first search on gpus.
\newblock In {\em SIGMOD}, pages 403--416, 2016.

\bibitem{luo2010effective}
L.~Luo, M.~Wong, and W.-m. Hwu.
\newblock An effective gpu implementation of breadth-first search.
\newblock In {\em DAC}, pages 52--55, 2010.

\bibitem{martone2010use}
M.~Martone, S.~Filippone, S.~Tucci, P.~Gepner, and M.~Paprzycki.
\newblock Use of hybrid recursive csr/coo data structures in sparse
  matrix-vector multiplication.
\newblock In {\em IMCSIT}, pages 327--335. IEEE, 2010.

\bibitem{McGregor:2014:GSA}
A.~McGregor.
\newblock Graph stream algorithms: A survey.
\newblock {\em SIGMOD Rec.}, 43(1):9--20, 2014.

\bibitem{Merrill:2015:HPS}
D.~Merrill, M.~Garland, and A.~Grimshaw.
\newblock {High-Performance and Scalable GPU Graph Traversal}.
\newblock {\em TOPC}, 1(2), 2015.

\bibitem{Murphy:2010:ITG}
R.~C. Murphy, K.~B. Wheeler, B.~W. Barrett, and J.~A. Ang.
\newblock Introducing the graph 500.
\newblock 2010.

\bibitem{ohsaka2015efficient}
N.~Ohsaka, T.~Maehara, and K.-i. Kawarabayashi.
\newblock Efficient pagerank tracking in evolving networks.
\newblock In {\em KDD}, pages 875--884, 2015.

\bibitem{saad1989numerical}
Y.~Saad.
\newblock Numerical solution of large nonsymmetric eigenvalue problems.
\newblock {\em Computer Physics Communications}, 53(1):71--90, 1989.

\bibitem{TweetStats}
D.~Sayce.
\newblock 10 billions tweets, number of tweets per day.
\newblock \url{http://www.dsayce.com/social-media/10-billions-tweets/}.
\newblock Accessed: 2016-10-18.

\bibitem{Soman:2010vm}
J.~Soman, K.~Kothapalli, and P.~J. Narayanan.
\newblock {A fast GPU algorithm for graph connectivity.}
\newblock {\em IPDPS Workshops}, 2010.

\bibitem{stonebraker20058}
M.~Stonebraker, U.~{\c{C}}etintemel, and S.~Zdonik.
\newblock The 8 requirements of real-time stream processing.
\newblock {\em ACM SIGMOD Record}, 34(4):42--47, 2005.

\bibitem{stratton2012optimization}
J.~A. Stratton, N.~Anssari, C.~Rodrigues, I.-J. Sung, N.~Obeid, L.~Chang, G.~D.
  Liu, and W.-m. Hwu.
\newblock Optimization and architecture effects on gpu computing workload
  performance.
\newblock In {\em InPar}, pages 1--10, 2012.

\bibitem{Tang:2016:GSS}
N.~Tang, Q.~Chen, and P.~Mitra.
\newblock Graph stream summarization: From big bang to big crunch.
\newblock In {\em SIGMOD}, pages 1481--1496, 2016.

\bibitem{Toshniwal:2014:STO}
A.~Toshniwal, S.~Taneja, A.~Shukla, K.~Ramasamy, J.~M. Patel, S.~Kulkarni,
  J.~Jackson, K.~Gade, M.~Fu, J.~Donham, N.~Bhagat, S.~Mittal, and D.~Ryaboy.
\newblock Storm@twitter.
\newblock In {\em SIGMOD}, pages 147--156, 2014.

\bibitem{Tsourakakis:2009:CTM}
C.~E. Tsourakakis, U.~Kang, G.~L. Miller, and C.~Faloutsos.
\newblock {DOULION:} counting triangles in massive graphs with a coin.
\newblock In {\em SIGKDD}, pages 837--846, 2009.

\bibitem{Verner:2012:SPR}
U.~Verner, A.~Schuster, M.~Silberstein, and A.~Mendelson.
\newblock Scheduling processing of real-time data streams on heterogeneous
  multi-gpu systems.
\newblock In {\em SYSTOR}, page~7, 2012.

\bibitem{Wang:2015:GHG}
Y.~Wang, A.~Davidson, Y.~Pan, Y.~Wu, A.~Riffel, and J.~D. Owens.
\newblock Gunrock: A high-performance graph processing library on the gpu.
\newblock {\em SIGPLAN Not.}, 50(8):265--266, 2015.

\bibitem{ywang}
Y.~Wang, Q.~Fan, Y.~Li, and K.-L. Tan.
\newblock Real-time influence maximization on dynamic social streams.
\newblock In {\em Proc. VLDB Endow.}, 2017.

\bibitem{yan2014yaspmv}
S.~Yan, C.~Li, Y.~Zhang, and H.~Zhou.
\newblock yaspmv: yet another spmv framework on gpus.
\newblock In {\em SIGPLAN Notices}, volume~49, pages 107--118, 2014.

\bibitem{Yang:2011:FSM}
X.~Yang, S.~Parthasarathy, and P.~Sadayappan.
\newblock Fast sparse matrix-vector multiplication on gpus: Implications for
  graph mining.
\newblock {\em Proc. VLDB Endow.}, 4(4):231--242, 2011.

\bibitem{yang2011fast}
X.~Yang, S.~Parthasarathy, and P.~Sadayappan.
\newblock Fast sparse matrix-vector multiplication on gpus: implications for
  graph mining.
\newblock {\em Proc. VLDB Endow.}, 4(4):231--242, 2011.

\bibitem{yang2016tracking}
Y.~Yang, Z.~Wang, J.~Pei, and E.~Chen.
\newblock Tracking influential nodes in dynamic networks.
\newblock {\em arXiv preprint arXiv:1602.04490}, 2016.

\bibitem{Zaharia:2013:DSF}
M.~Zaharia, T.~Das, H.~Li, T.~Hunter, S.~Shenker, and I.~Stoica.
\newblock Discretized streams: Fault-tolerant streaming computation at scale.
\newblock In {\em SOSP}, pages 423--438, 2013.

\bibitem{zhang2016approximate}
H.~Zhang, P.~Lofgren, and A.~Goel.
\newblock Approximate personalized pagerank on dynamic graphs.
\newblock {\em arXiv preprint arXiv:1603.07796}, 2016.

\bibitem{Zhang:2011:GStream}
Y.~Zhang and F.~Mueller.
\newblock Gstream: {A} general-purpose data streaming framework on {GPU}
  clusters.
\newblock In {\em ICPP}, pages 245--254, 2011.

\bibitem{Zhong:2014:SGP}
J.~Zhong and B.~He.
\newblock Medusa: Simplified graph processing on gpus.
\newblock {\em {IEEE} Trans. Parallel Distrib. Syst.}, 25(6):1543--1552, 2014.

\end{thebibliography}

\end{document}